\documentclass[12pt]{article}

\usepackage[left=1in,top=1.2in,bottom=1in,right=1in]{geometry} 
\usepackage{times}
\usepackage{url}
\usepackage{algorithm, algorithmic, amscd, amsfonts, amsmath, amssymb}
\usepackage{amstext, amsthm, bbm, bm, booktabs, caption, color, footnote}
\usepackage{graphics, graphicx, lineno, multirow, rotating, ragged2e}
\usepackage{setspace, subfigure, ulem, setspace}

\newtheorem{theorem}{Theorem}[section]
\newtheorem{lemma}[theorem]{Lemma}

\makeatletter
\newcommand{\vast}{\bBigg@{4}}
\newcommand{\Vast}{\bBigg@{5}}

\newcommand{\eat}[1]{}

\newcommand{\ie}{{\it i.e.}}

\newcommand{\minimize}[1]{\underset{#1}{\mbox{minimize}}}

\newcommand{\argmin}[1]{\underset{#1}{\mbox{argmin}}}
\newcommand{\argmax}[1]{\underset{#1}{\mbox{argmax}}}

\newcommand{\st}{\mbox{subject to}}
\newcommand{\la}{\langle}
\newcommand{\ra}{\rangle}
\newcommand{\vect}{\mbox{vec}}
\definecolor{dgreen}{rgb}{0,0.5,0}

\makeatother

\allowdisplaybreaks

\title{G-AMA: Sparse Gaussian graphical model estimation via alternating minimization}

\author{ \large{Onkar Anant Dalal} \\
 \textit{{Institute for Computational \& Mathematical Engg,}}\\
 \textit{{Stanford University, USA}}\\
 \textit{onkar@alumni.stanford.edu}\\
\large{Bala Rajaratnam} \\
\textit{{Department of Statistics,}} \\
\textit{{Stanford University, USA}} \\
\textit{brajarat@stanford.edu}
} 
\date{}
\begin{document}

\pagenumbering{gobble}
\vspace*{-70pt}
   {\let\newpage\relax\maketitle}

\begin{abstract}
Several methods have been recently proposed for estimating sparse Gaussian graphical models using $\ell_{1}$ regularization on the inverse covariance matrix. Despite recent advances, contemporary applications require methods that are even faster in order to handle ill-conditioned high dimensional modern day datasets. In this paper, we propose a new method, \text{G-AMA}, to solve the sparse inverse covariance estimation problem using Alternating Minimization Algorithm (\text{AMA}), that effectively works as a proximal gradient algorithm on the dual problem. Our approach has several novel advantages over existing methods. First, we demonstrate that \text{G-AMA} is faster than the previous best algorithms by many orders of magnitude and is thus an ideal approach for modern high throughput applications. Second, global linear convergence of \text{G-AMA} is demonstrated rigorously, underscoring its good theoretical properties. Third, the dual algorithm operates on the covariance matrix, and thus easily facilitates incorporating additional constraints on pairwise/marginal relationships between feature pairs based on domain specific knowledge. Over and above estimating a sparse inverse covariance matrix, we also illustrate how to (1) incorporate constraints on the (bivariate) correlations and, (2) incorporate equality (equisparsity) or linear constraints between individual inverse covariance elements. Fourth, we also show that \text{G-AMA} is better adept at handling extremely ill-conditioned problems, as is often the case with real data. The methodology is demonstrated on both simulated and real datasets to illustrate its superior performance over recently proposed methods.
\end{abstract}

 
\section{Introduction}
\label{s:intro}
In this paper, we consider the problem of sparse inverse covariance estimation for undirected Gaussian graphical models using $\ell_{1}$ regularized maximum likelihood estimation (MLE). Given $n$ realizations of a $p$ dimensional Gaussian random vector, with population and sample covariance matrix denoted by $S$ and $\Sigma$ respectively, the goal is to estimate $\Sigma \in \bm{S}_{++}^{p}$ so that $\hat{\Sigma}^{-1}$ is sparse. For the multivariate Gaussian distribution, the sparsity in the inverse covariance is related to the conditional independence among the features. Specifically, two features $i$ and $j$ are conditionally independent, if and only if the corresponding entry in the inverse covariance matrix is zero, i.e.,$\Sigma^{-1}_{ij} = 0$. Since the MLE is formulated in terms of the inverse covariance $X$, adding an $\ell_{1}$ regularization with a penalty parameter $\lambda$ induces sparsity in the estimated inverse covariance matrix. The regularized maximum likelihood estimation problem in terms of $\Sigma^{-1} = X$ is given by 
\begin{alignat}{1}
\label{mainprob}
\minimize{X} \ & - \log\det X  + \langle S, X \rangle + \lambda \|X\|_{1} \nonumber \\
\st & \ \ X \in \bm{S}_{++}^{p}
\end{alignat}
where $\|X\|_{1} = \sum_{i,j} |X_{ij}|$. The optimization problem (\ref{mainprob}) is a convex minimization problem in $X$. For $\lambda = 0$, the solution of Problem (\ref{mainprob}), given by $X^{(0)} = (\hat{\Sigma}^{(0)})^{-1} = S^{-1}$ is an unbiased estimator of the covariance matrix, but is not well defined for $n < p$. However, for any $\lambda > 0$ the estimate is well defined for any $n$, with sparsity in $\hat{\Sigma}^{-1}$ increasing for higher values of $\lambda$. 

The dual of problem (\ref{mainprob}) can be formulated in terms of the covariance matrix $\Sigma = Y$ as 
\begin{alignat}{1}
\label{maindual}
\minimize{Y} \ & -\log\det Y - p\nonumber \\
\st & \  \  \| Y - S \|_{\infty} \leq \lambda.
\end{alignat}
In this paper, we propose a new algorithm called Graphical-Alternating Minimization Algorithm (\text{G-AMA}), which solves (\ref{mainprob}) using the Alternating Minimization Algorithm (\text{AMA}) proposed by \cite{tseng1991applications}. In \cite{tseng1991applications}, \text{AMA} was shown to solve the dual problem with the Forward-Backward Splitting (\text{FBS}) method from \cite{rockafellar2011variational} and therefore converges linearly for strongly monotone operators. The iterates of \text{G-AMA} are always maintain a feasible dual estimate for the covariance matrix $\hat{\Sigma}$. This is specifically useful while solving large problems with budgeted time and early termination. In addition, \text{G-AMA} is very fast for very small values of $\lambda$ where other state-of-the-art methods converges slowly in comparison since the solutions are highly ill-conditioned. 

In many modern practical applications, domain specific knowledge on the covariance structure is known but is not always utilized to its full extent. This knowledge is often available on the covariance matrix itself, as compared to the inverse covariance matrix, since the sample covariance matrix $S$ is still computable even in the sample starved setting albeit being only positive semi-definite. There are many such practical examples. In genomics it is well known that pairwise correlations of gene (expressions) with genes further downstream are  smaller in magnitude. In the environmental sciences when modeling spatial fields, it is often the case that correlations between distant spatial points are less pronounced. As a relative by-product of our \text{AMA} analysis, we investigate how to modify \text{G-AMA} and other methods to incorporate bound constraints on bivariate covariances if such domain specific information is readily available. We demonstrate that the \text{G-AMA} formulation facilitates such analysis quite easily.

The remainder of the paper is organized as follows: Section 1 provides a brief survey of prior work for solving the sparse inverse covariance estimation problem. The methodology and details of \text{G-AMA} are given in Section 2, followed by connection with other approaches. Section 3 considers convergence analysis of \text{G-AMA} and Section 4 extends the results in Section 3 to two generalizations of the main problem. A detailed account of numerical experiments on synthetic as well as real-life datasets is given in Section 5. This section also demonstrates the use of sparse inverse covariance for portfolio optimization. Further details of the \text{G-AMA} method, proofs and additional details on the portfolio optimization application are provided in the appendix.
 
\subsection{Prior Work}
A number of algorithms have been proposed to solve the primal problem (\ref{mainprob}) and its dual problem (\ref{maindual}). These algorithms can be briefly classified into two classes, namely, the block-coordinate descent methods and the proximal gradient/Newton methods. The proposed method \text{G-AMA} belongs to the second class. In this subsection, we briefly describe algorithms that have been proposed thus far. 

 
\subsection*{Block Coordinate Descent Methods}
A dual block coordinate descent which solves a box-QP for each coordinate was first proposed in \cite{banerjee2008model}. In \cite{friedman2008sparse}, Friedman et al. noticed that each step in dual of the box-QP is equivalent to solving a lasso problem, and proposed the \text{glasso} algorithm. The block coordinate descent in \text{glasso} was improved upon by applying similar techniques to the primal problem as \text{p-glasso} and \text{dp-glasso} in \cite{mazumder2012graphical}. These algorithms take a coordinate descent step for each row and corresponding column and iterate cyclically until convergence. These algorithms are shown to converge to the optimal primal or dual variable. However, the rate of convergence of these methods have not been established, thus providing no theoretical guarantees for the speed of convergence.

 
\subsection*{Proximal Gradient and Proximal Newton-like Methods}
In \cite{d2008first}, \cite{lu2009smooth}, and \cite{lu2010adaptive} the authors use Nesterov's smooth approximation methods and its variations to propose algorithms which achieve $\epsilon$-convergence in ${\cal O}(1/\epsilon)$ and ${\cal O}(1/\sqrt\epsilon)$ respectively. A proximal method, \text{QUIC}, which uses a proximal Newton step was proposed in \cite{hsieh2011sparse}. The QUIC algorithm is guaranteed to converge Q-quadratically provided the iterates are close enough to the optimal. This algorithm was later generalized in \cite{lee2012proximal}. In recent work \cite{brajarat2012nips}, a proximal gradient algorithm, \text{G-ISTA}, was proposed and shown to achieve global linear convergence. Both QUIC and G-ISTA algorithms have their respective advantages in various sparsity and condition number regimes. Moreover, both \text{G-ISTA} and \text{QUIC} use the inverse covariance matrix as their operating variable and seem to outperform the coordinate descent glasso methods in numerical examples.

\section{\text{G-AMA}: An alternating minimization algorithm for sparse graphical model estimation}
 
 
\subsection{Methodology and problem formulation}
In this section, we describe our proposed methodology to solve Problem (\ref{mainprob}) using alternating minimization algorithm \cite{tseng1991applications}. We write problem (\ref{mainprob}) is composite form using a dummy variable $Z$, with $h(X) = - \log\det X$, and $g(Z) = \langle S, Z \rangle + \lambda \|Z\|_{1}$ as 
\begin{alignat}{1}
\label{mainamaprob}
\minimize{X,Z} \ & - \log\det X + \langle S, Z \rangle + \lambda \|Z\|_{1}, \nonumber \\
\st \ & Z - X = 0.
\end{alignat}
An \text{AMA} iteration updates the primal variables $X$, $Z$ and the dual variable $Y$ sequentially using 
\begin{alignat}{1}
\label{subproblems}
	X_{+} &= \argmin{X} \,\,\, - \log\det X - \langle Y, -X \rangle, \nonumber \\
	Z_{+} &= \argmin{Z} \,\,\, \langle  S, Z \rangle + \lambda \|Z\|_{1} - \langle Y, Z \rangle + \frac{\tau}{2}\| X_{+} - Z \|_{F}^{2}, \nonumber \\
	Y_{+} &= Y + \tau(X_{+} - Z_{+}),
\end{alignat}
where $\tau$ is the step-size. A key feature of the optimization problems in the first two steps is that they can be solved analytically to obtain closed form expressions. The optimality conditions for these two problems are
\begin{alignat}{1}
0 &= -X_{+}^{-1} + Y, \nonumber \\
0 &\in S + \lambda\mbox{Sign}(Z_{+}) - Y + {\tau}(Z_{+} - X_{+}). \nonumber
\end{alignat}
The second optimality condition can be rewritten using the soft-thresholding function ${\cal S}_{\lambda}(x) = \mbox{sign}(x)\left(\max \{(|x| - \lambda),0 \}\right)$ as 
\begin{alignat}{1}
Z_{+} & = {\cal S}_{\lambda/\tau} (X_{+} + (Y - S)/\tau), \nonumber
\end{alignat}
where the soft-thresholding operator is applied entry-wise. In addition, we observe that substituting $X_{+}$ and $Z_{+}$ in terms of $Y$ yields a one step update for $Y_{+}$ which can be written as
\begin{alignat}{1}
\label{gamaiter}
Y_{+} &= \min \left\{ \max \left\{ Y - S + \tau Y^{-1}, -\lambda \right\}, \lambda \right\} + S.
\end{alignat}
We define a clip function as ${\cal C}_{\lambda}(x) = \min\left\{\max\left\{x, -\lambda\right\}, \lambda\right\}$, where the $\min$, $\max$ functions clip $x$ between $-\lambda$ and $\lambda$. The clip function is related to the soft-thresholding function via the identity 
\begin{alignat}{1}
x &= {\cal S}_{\lambda}(x) + {\cal C}_{\lambda}(x). \nonumber
\end{alignat} 
This identity is used in the next section for proof of convergence. The details for alternating minimization algorithm for sparse graphical model estimation (\text{G-AMA}) are given in Algorithm \ref{a:gama}.
\begin{algorithm}
\caption{\text{G-AMA}: Graphical - Alternating Minimization Algorithm}
\label{a:gama}
\begin{tabbing}
\enspace \text{input}: sample covariance $S$, regularization $\lambda$, tolerance $\epsilon_{\text{opt}}$, backtracking constant $c \in (0,1)$. \\
\enspace \text{initialize}: $Y_{0}$, $\tau_{0,0} = 1$, $\Delta_{\text{opt}} = 2\epsilon_{\text{opt}}$. \\
\enspace \text{while}: {$\Delta_{\text{opt}} > \epsilon_{\text{opt}}$}, \\
	\qquad $X_{k+1} = Y_{k}^{-1}$,\\
	\qquad \text{compute} $\tau_{k}$: Largest feasible step in $\{c^{j}\tau_{k,0}\}_{j=0,1,\ldots}$ such that $Y_{k+1} \succ 0$ and satisfies (\ref{suffdesc}),\\
	\qquad $Z_{k+1} = {\cal S}_{\lambda/\tau_{k}}(X_{k+1} + (Y_{k}-S)/\tau_{k}))$,\\
	\qquad $Y_{k+1} = Y_{k} + \tau_{k}(X_{k+1} - Z_{k+1})$,\\
	\qquad $\tau_{k+1,0} = \displaystyle\frac{\la Y_{k+1}-Y_{k}, Y_{k+1}-Y_{k} \ra}{\la Y_{k+1}-Y_{k}, X_{k}-X_{k+1} \ra}$,\\
	\qquad $\Delta_{\text{opt}} = -\log\det Y_{k+1} - p - \log\det X_{k+1} + \langle S, X_{k+1} \rangle + \lambda\|X_{k+1}\|_{1}$,\\
\enspace \text{endwhile}\\
\enspace \text{output}: $\epsilon$-optimal solution to problem (\ref{mainprob}), covariance estimate $\hat{\Sigma}^{(\lambda)} = Y_{k+1}$.\\
\end{tabbing}
\end{algorithm}

The Algorithm \ref{a:gama} is terminated using an $\epsilon_{\texttt{opt}}$ tolerance on the duality gap. The step size for each iteration is chosen using the Barzilia-Borwein (BB) step (see \ref{stepsize} in Supplemental section) which is a two point approximation to the secant equation \cite{barzilai1988two}. A backtracking line search is conducted such that the next iterate is positive definite and satisfies a sufficient descent given by 
\begin{alignat}{1}
\label{suffdesc}
-\log\det(Y_{k+1}) \leq -\log\det(Y_{k}) + \la Y_{k+1}-Y_{k}, Y_{k}^{-1} \ra + \frac{1}{2\tau}\|Y_{k+1}-Y_{k}\|^{2}_{F},
\end{alignat}
where the right hand side is a local quadratic approximation of the dual objective around $Y_{k}$. In case a number of backtracking steps fail to satisfy these two criteria, a safe step of $\tau = \alpha^{2}$ is taken based on Lemma \ref{lowerbound}. More details about step size selection are given in the Supplemental section.

For \text{G-AMA}, the objective function of the dual problem (\ref{maindual}) can be written as $h^{*}(Y) + g^{*}(Y)$ with
\begin{alignat}{1}
\label{dualobjectives}
h^{*}(Y) = -\log\det Y, \ \ \ \ \mbox{and} \ \ \ \ g^{*}(Y) = \bm{1}_{\{\|Y-S\|_{\infty} \leq \lambda\}}.
\end{alignat}
The gradient and hessian of the smooth dual objective $h^{*}$ are given by
\begin{alignat}{1}
\nabla{h}^{*}(Y) = -Y^{-1}, \nonumber
\end{alignat} 
and
\begin{alignat}{1}
\nabla^{2}{h}^{*}(Y) = -Y^{-1} \otimes Y^{-1}. \nonumber
\end{alignat} 
In the following section, we show that over any compact subset of ${\cal S}_{++}^{p}$ the function $h^{*}$ is strongly convex and the gradient $\nabla{h^*}$ is Lipschitz continuous. These properties are useful in establishing global linear convergence of \text{G-AMA} using properties of \text{AMA} from \cite{tseng1991applications}.  

 
\section{Convergence analysis and rate of convergence}
In this section, we prove theoretical results regarding global convergence of \text{G-AMA}. We first show strong convexity and Lipschitz continuity of the gradient of the objective function over any compact domain. Next we show that the iterates of \text{G-AMA} belong to a compact domain bounded away from the boundary of the positive definite cone $S_{++}^{p}$. Finally we show the linear convergence of the \text{G-AMA} iterates. The proofs are given in the Supplementary section.

We begin with the optimality conditions for Problem (\ref{mainprob}). The subgradient condition gives
\begin{alignat}{1}
- X_{*}^{-1} + S - \lambda \mbox{Sign}(X_{*}) \ni 0. \nonumber
\end{alignat}
This translates to 
\begin{alignat}{1}
\label{optcond}
X_{*}(i,j) = 0 \iff & |S(i,j) - X_{*}^{-1}(i,j) | \leq \lambda \,\,\,\, \mbox{and} \nonumber\\
X_{*}(i,j) \neq 0 \iff & S(i,j) - X_{*}^{-1}(i,j)  = -\lambda\cdot\mbox{Sign}\{X_{*}(i,j)\}. 
\end{alignat}
We now show that the optimal point satisfying the above conditions is a fixed point for the \text{G-AMA} iterations and vice versa.
\begin{lemma}
\label{fixedpointlemma}
A matrix $X_{*}$ is the optimal solution of the Problem (\ref{mainprob}) satisfying (\ref{optcond}) if and only if the inverse $Y_{*} = X_{*}^{-1}$ is a fixed point of the \text{G-AMA} iteration in (\ref{gamaiter}), i.e.,
\begin{alignat}{1}
\label{fixedpoint}
Y_{*} &= {\cal C}_{\lambda}\left(\tau Y_{*}^{-1} + (Y_{*} - S) \right) + S.
\end{alignat}
\end{lemma}
\begin{proof}
The proof is given in Appendix.
\end{proof}
The existence of fixed point (\ref{fixedpoint}) will allow us to exploit arguments similar to those in \cite{brajarat2012nips} to prove global linear convergence. 

Now, note that the gradient $\nabla\log\det Y = Y^{-1}$ is not Lipschitz continuous over the entire domain of $\bm{S}_{++}^{p}$. However, the gradient can be shown to be Lipschitz continuous over the compact domain  
\begin{alignat}{1}
\label{compact}
{\cal D} &= \left\{ Y \mid \alpha I \preceq Y \preceq \beta I \right\} \subset \bm{S}_{++}^{p}, \ \ \ \mbox{for} \ \ \ 0 < \alpha < \beta < \infty.
\end{alignat} 
\begin{lemma}
\label{lipschitz}
(\cite[Lemma 2]{brajarat2012nips}). For $Y_{1}, Y_{2} \in \bm{S}_{++}^{p}$, the gradient $\nabla\log\det Y = Y^{-1}$ satisfies 
\begin{alignat}{1}
\displaystyle\frac{1}{\beta^{2}} \| Y_{1} - Y_{2} \|_{2} \leq \| Y_{1}^{-1} - Y_{2}^{-1} \|_{2} \leq \frac{1}{\alpha^{2}} \| Y_{1} - Y_{2} \|_{2}, \nonumber
\end{alignat}
where $\alpha = \min\left\{ \lambda_{\min}(Y_{1}), \lambda_{\min}(Y_{2}) \right\}$ and $\beta = \max\left\{ \lambda_{\max}(Y_{1}), \lambda_{\max}(Y_{2}) \right\}$.
\end{lemma}
We also note that the hessian of $\log\det$ function is given by
\begin{alignat}{1}
\label{strongconvex}
\nabla^{2}\log\det Y = -Y^{-1} \otimes Y^{-1},
\end{alignat} 
and the function $\log\det Y$ is strongly convex when $Y$ is restricted to the domain ${\cal D}$. Next we show that the iterates of \text{G-AMA} belong to the bounded compact set ${\cal D}$ defined in (\ref{compact}) and give explicit values for the constants $\alpha$ and $\beta$. We begin with the upper bound based on \cite[Lemma 8]{brajarat2012nips}.
\begin{lemma}
\label{upperbound}
(\cite[Lemma 8]{brajarat2012nips}). Let $0 \prec \alpha_{l} I \preceq Y_{l}$ for $l = 0, 1 \ldots k$ and let $\tau_{k} < \alpha_{k}^{2} < \beta^{2}$ (for $\beta =  \|Y_{0}-Y_{*}\|_{F} + \|Y_{*}\|_{2}$) be the step size used for \text{G-AMA} iterations starting with $\alpha_{0} I \prec Y_{0}$. Then the next iterates $Y_{k+1}$ satisfies
\begin{alignat}{1}
\label{betadef}
Y_{k+1} &\preceq \beta I,
\end{alignat}
and hence by inductive argument $Y_{k} \preceq \beta I$, for all $k$.
\end{lemma}
\begin{proof}
The proof is given in Appendix.
\end{proof}
Note that Lemma \ref{upperbound} only assumes a strictly positive $\alpha_{l}$ for $l = 0, 1,\ldots,k$ to prove the upper bound on $Y_{k+1}$. It does not assume that the subsequent iterates are strictly bounded away from the semidefinite boundary. This result is shown in the next lemma which will use a universal upper bound $\beta$ on $\|Y_{k}\|_{2}$, for all $(k = 0, 1, \ldots)$. Based on Lemma 3 in \cite{hsieh2011sparse} we show that the iterates of \text{G-AMA} are bounded away from the boundary of the positive definite cone and satisfy $\alpha I \preceq Y_{k}$ for some universal $\alpha > 0$ for all $(k = 0, 1, \ldots)$.
\begin{lemma}
\label{lowerbound}
(\cite[Lemma 3]{hsieh2011sparse}). The set ${\cal U} = \left\{ Y \in \bm{S}_{++}^{p} \mid Y\prec \beta I \,\,\, \mbox{and} \,\,\, \log\det Y_{0} < \log\det Y \right\}$ satisfies $\alpha I \preceq Y$ with $\alpha = \frac{\lambda^{p}}{\beta^{(p-1)}}$, where $\beta$ is the upper bound from (\ref{betadef}) of Lemma \ref{upperbound}.
\end{lemma}
\begin{proof}
The proof is given in Appendix.
\end{proof}

We now show global linear convergence of \text{G-AMA} using Lemma \ref{upperbound}, and Lemma \ref{lowerbound}.
\begin{theorem}
\label{t:gama}
The iterates $Y_{k}$ of Algorithm \ref{a:gama} satisfy $\alpha I \preceq Y_{k} \preceq \beta I$ as well as 
\vspace{-0.02in}\begin{alignat}{1}
\label{mt1}
\|Y_{k+1} - Y_{*}\|_{F} \leq \displaystyle\max\left\{\bigg|1-\frac{\tau_k}{\alpha^{2}}\bigg|, \bigg|1-\frac{\tau_k}{\beta^{2}}\bigg|\right\} \|Y_{k} - Y_{*}\|_{F}, \vspace{-0.02in}
\end{alignat}
\vspace{-0.02in} for $k = 0, 1 \ldots$. More specifically, for step size $\tau_{k} < \alpha^{2}$,
\begin{alignat}{1}
\label{mt2}
\|Y_{k+1} - Y_{*}\|_{F} \leq \gamma \|Y_{k} - Y_{*}\|_{F},
\end{alignat}
where $\gamma < 1$ and hence the iterates converges to a $\epsilon$-optimal solution $Y_{*}$ in ${\cal O}\left(\log(1/\epsilon)\right)$ iterations.
\end{theorem}
\begin{proof}
The proof is given in Appendix.
\end{proof}

The maximum step size $\alpha^2$ provided by $\alpha$ from Lemma \ref{lowerbound} is very conservative. We show later that much better performance can be achieved in practice using heuristics for choosing the step size.

 
\subsection{Connections to \text{G-ISTA}} 
The dual objective functions $h^{*}$ and $g^{*}$ from (\ref{dualobjectives}) are conjugate functions (see \cite{boyd2004convex}) of the primal objective functions 
\begin{alignat}{1}
h(X) = -\log\det X, \ \ \ \ \mbox{and}  \ \ \ \ g(Z) = \la S, Z \ra + \|Z\|_{1}. \nonumber
\end{alignat} 
In \cite{brajarat2012nips}, Problem (\ref{mainprob}) is solved using a proximal gradient method (\text{FBS} also known as \text{ISTA}) and global linear convergence of the algorithm is established using strong convexity of $h$ and Lipschitz continuity of the gradient, $\nabla{h}$, over a compact domain. Given an initial point, the subsequent iterates are shown to remain in a compact domain and hence the proximal gradient method converges linearly \cite{rockafellar1976monotone}. The two algorithms \text{G-AMA} and \text{G-ISTA} are closely related, and as expected there are some parallels in their respective convergence analysis. The parallels between the two methods can also be used to choose the step size using an equivalent Barzilia-Borwein step.

 
\subsection{Connections to \text{QUIC} and \text{DC-QUIC}} 
An extension wrapper for QUIC called DC-QUIC was proposed in \cite{hsieh2012divide} that provides a divide-and-conquer framework which approximates a large inverse covariance estimation problem with multiple small problems. These problems are then solved using standard QUIC and the solutions are concatenated to get an approximate solution of the original problem. More specifically, the DC wrapper is an independent top layer that only uses QUIC as its base algorithm. Any equivalent method like G-AMA can replace QUIC in this bigger divide-and-conquer framework. In this light, our current work compares the base algorithms of G-AMA and QUIC, and the speedup of G-AMA over QUIC can therefore be also directly utilized in the spirit of a ``DC-GAMA". 

\subsection{Connections to \text{ALM} \cite{scheinberg2010sparse}} 
We also note that the G-AMA algorithm is different from the ALM method \cite{scheinberg2010sparse} in several fundamental aspects, including (a) Theoretical, (b) Algorithmic and, (c) Numerical. In particular, (a) The ALM algorithm is shown to attain an $\epsilon$-optimal solution in O(1/$\epsilon$) iterations. In comparison, G-AMA is shown to attain an $\epsilon$-optimal solution in O($\log$(1/$\epsilon$)) iterations. The rate of convergence of our proposed method G-AMA is therefore faster, and it is not faster just by an order of magnitude (such as O(1/$\epsilon^2$) of Nesterov's accelarated gradient methods), it is in a different regime altogether. In particular, G-AMA has linear convergence. (b) The ALM algorithm is an implementation of ADMM \cite{boyd2011distributed}, whereas G-AMA is an implementation of AMA \cite{tseng1991applications}. The key difference between these methods is the first subproblem in (\ref{subproblems}) where the ADMM algorithm uses a quadratic term and AMA does not. As we shall see this very vital difference provides the speedup enjoyed by G-AMA as a result of the strongly convex objective function. (c) Numerically, ALM was substantially shown to be slower than QUIC in \cite{hsieh2011sparse} and was therefore not considered state-of-the-art in comparisons with G-AMA. To be fair, G-AMA will be compared to the more superior QUIC and G-ISTA algorithms and will be shown to outperform these current state-of-the-art methods.

\subsection{Choice of penalty parameter and Statistical Consistency}
Once convergence of G-AMA to a global minimum is established, the next natural question to investigate is the choice of penalty parameter $\lambda$. There are various approaches that have been proposed in this regard. First, cross validation is a popular technique that is often used. Second, one could use a value for the penalty parameter $\lambda$ that yields a desired level of sparsity (such as 3\% edge density). Third, Bayesian information type criteria are also popular for selection of the appropriate level of regularization \cite{sang2013concord}. Fourth, the penalty parameter can be chosen so as to control the probability of false detection of edges, i.e., a certain predetermined level of error control \cite{banerjee2008model}. An important question pertains to asymptotic recovery of the underlying graph when the data is generated from a sparse Gaussian graphical model (i.e., model selection consistency). Since G-AMA yields the $\ell_1$ regularized MLE (similar to G-lasso, G-ISTA and other approaches), we can leverage already established statistical theory to assert asymptotic consistency of G-AMA estimates. The reader is referred to \cite{banerjee2008model} in this regard, who show that  provided the dimension $p = {\cal O}(n^{\gamma})$ for some $\gamma > 0$, the $\ell_1$ regularized MLE recovers the underlying graph with probability tending to 1.

\section{Generalized \text{G-AMA}}
\label{s:gengama}
We now proceed to demonstrate that the G-AMA framework is even richer than the formulation in Problem (\ref{mainamaprob}). In this section, we extend our \text{G-AMA} algorithm to solve two natural generalizations of the sparse covariance estimation problem.

The first generalization modifies the dual problem (\ref{maindual}) to include explicit constraints on the bivariate covariances in $\Sigma$. As discussed in Section \ref{s:intro}, scenarios where domain specific knowledge regarding correlation structure is available arise in numerous applications. This knowledge/structure is more easily found in bivariate pairwise/marginal relationships between feature pairs since such quantities can be simply calculated from the sample covariance matrix, as compared to the sample inverse covariance matrix which is not defined when $n < p$. Since the operating variable for \text{G-AMA} is the estimate of the covariance matrix $\hat{\Sigma} = Y$, it can easily incorporate this domain specific knowledge by adding constraints on the covariance matrix (compared to the inverse covariance matrix). Thus the G-AMA framework allows for an additional layer of regularization, while maintaining the sparsity in the inverse covariance matrix, and in the process enriches the type of graphical modeling that can be achieved. We develop a comprehensive methodology in this regard that extends the standard inverse covariance estimation problem. The dual method for \text{glasso} can also be modified to incorporate simple bound constraints like $l_{ij} \leq \sigma_{ij} \leq u_{ij}$. However, a more complicated constraint would lose the fast coordinatewise descent algorithm for the lasso subproblem.

The second generalization modifies the primal problem (\ref{mainprob}) with a $\ell_{1}$ regularization on a symmetric linear transform of the inverse covariance matrix $X$ (see Problem (\ref{mod2prim})). In many applications the underlying graph structure is known to be symmetric, i.e., the edge weights (partial covariance/correlations) assume equal values (equisparsity) either naturally or by design (see \cite{gehrmann2012estimation} and references therein for examples). In such settings, it is often both necessary and useful to estimate a partial correlation graph that respects these symmetries (equisparisty). Such additional constraints on graphical models are both of theoretical and applied interest due to the added layer of regularization and structure it provides. This type of information about the structure of the underlying graph can be incorporated by using a regularization term in the form of a symmetric linear transform.

\subsection{Methodology}
\subsection*{Generalization of convex constraint on $\Sigma$}
\label{ss:gen1}
The constraint set on $\Sigma$ in Problem (\ref{maindual}) can be generalized to any convex constraint ${\cal D}$. In this subsection we will illustrate the extension of \text{G-AMA} for a problem with arbitrary bound constraints on the bivariate covariances. We also show that constraints on the covariance matrix can sometimes be translated directly to constraints on the inverse covariance, hence primal methods can also be somewhat useful in this context. The dual problem can be modified with generalized bound constraints as
\begin{alignat}{1}
\label{moddual}
\minimize{Y} \ & -\log\det Y \nonumber \\
\st \ & \,\,\, l_{ij} \leq Y(i,j) \leq u_{ij}.
\end{alignat}

Here, each bivariate covariance $\sigma_{ij}$ is allowed to vary in the interval $[l_{ij}, u_{ij}]$ which is a generalization of the interval $[s_{ij}-\lambda,s_{ij}+\lambda]$ from the dual of regularized MLE problem (\ref{maindual}). These modified constraints break away from the maximum likelihood framework and allows a different kind of regularization based on domain specific knowledge. The corresponding primal problem formulated in terms of the inverse covariance matrix is given by
\begin{alignat}{1}
\label{modprim}
\minimize{X} \ & - \log\det X  + \langle \bar{S}, X \rangle + \sum_{i,j} \bar{\lambda}_{ij}|X_{ij}| \nonumber \\
\st \ & \,\,\, X \in \bm{S}_{++}^{p}
\end{alignat}
where $\bar{S}$ is a modified sample covariance (midpoints of the intervals) and the third term is a $\ell_{1}$ penalty with more flexible regularization $\bar{\lambda}_{ij}$ (half length of the intervals) for $X_{ij}$. These are related to the lower and upper bounds of the dual formulation as
\begin{alignat}{1}
{\bar{S}}_{ij} = \frac{l_{ij} + u_{ij}}{2}, \,\,\, \mbox{and} \,\,\, {\bar{\lambda}}_{ij} = \frac{u_{ij}-l_{ij}}{2}.
\end{alignat}

Note that unlike the MLE formulation, where $S \succeq 0$ provided a feasible point $S + \lambda I \succ 0$ (for any $\lambda > 0$), the generalized problem is no longer guaranteed to be feasible for arbitrary choice of parameters $l_{ij}$'s and $u_{ij}$'s. 
For this problem, the algorithm update for $Y$ is given by
\begin{alignat}{1}
\label{boundconupdate}
Y_{+} &= \Pi_{\cal D}\left\{ Y + \tau {Y}^{-1} \right\},
\end{alignat}
where $\Pi_{\cal D}$ is the projection onto the convex constraint set ${\cal D} = \left\{ x \ | \ l_{ij} \leq x_{ij} \leq u_{ij}, \forall{i,j} \right\}$ defined by the bound constraints. The simplified projection for the bound constraints defined in (\ref{moddual}) is given by
\begin{alignat}{1}
\{Y_{+}\}_{ij} &= \min\left\{\max\left\{ \left(Y + \tau {Y}^{-1}\right)_{ij}, l_{ij} \right\}, u_{ij}\right\}. \nonumber
\end{alignat}
The methodology for choosing the step size remains the same as for the standard G-AMA algorithm and the linear convergence results hold.

It is important to note that the generalization to any convex constraint set ${\cal D}$, in addition to the bound constraints or the standard constraint set of $\bm{1}_{\{\|Y-S\|_{\infty}\leq \lambda\}}$, might not always be easily translated to the primal problem. This fact underscores the flexibility of the above generalized G-AMA formulation on the dual.
\subsection*{Generalization of $\ell_{1}$ penalty on $\Omega$}
\label{ss:gen2}
Consider a generalized regularization on the $\ell_1$ norm of $A^{*}XB$ instead of the $\ell_1$ norm of $X$. Then a generalized version of problem (\ref{mainprob}) is given by
\begin{alignat}{1}
\label{mod2prim}
\minimize{X} \ & -\log\det X + \la S, X\ra + \lambda \|A^{*}XB\|_{1} \nonumber \\
\st \ & \,\,\, X \succ 0.
\end{alignat}
It can be easily shown that the above formulation enforeces equality (equisparsity or linear constraints) on the elements of the inverse covariance matrix. As in (\ref{mainamaprob}), the problem can be decomposed into composite form using the constraint $A^{*}XB - Z = 0$. Applying the alternating minimization method, yields the following update steps
\begin{alignat}{1}
	X_{+} &= \argmin{X} \,\,\, - \log\det X + \la S, Z \ra - \la Y, -A^{*}XB) \ra, \nonumber \\
	Z_{+} &= \argmin{Z} \,\,\,  \lambda \|Z\|_{1} - \langle Y, Z \rangle + \frac{\tau}{2}\| Z - A^{*}X_{+}B \|_{F}^{2}, \nonumber \\
	Y_{+} &= Y + \tau\left(A^{*}X_{+}B - Z_{+}\right).
\end{alignat}
The optimality conditions for the first two steps are given by
\begin{alignat}{1}
0 &= -X_{+}^{-1} + S + BYA^{*}, \nonumber \\
0 &\in \lambda\mbox{sign}(Z_{+}) - Y + {\tau}\left(Z_{+} - A^{*}X_{+}B\right).
\end{alignat}
The second optimality condition can be rewritten using the soft-threshold function as 
\begin{alignat}{1}
Z_{+} & = {\cal S}_{\lambda/\tau} \left(A^{*}X_{+}B + Y/\tau\right).
\end{alignat}
Substituting the two optimality conditions for $X_{+}$ and $Z_{+}$ gives a direct update for $Y_{+}$ in terms of $Y$ as
\begin{alignat}{1}
\label{mod2iter}
Y_{+} &= Y + \tau\left[A^{*}X_{+}B - {\cal S}_{\lambda/\tau} \left(A^{*}X_{+}B + Y/\tau\right) \right] \nonumber \\
      &= {\cal C}_{\lambda} \left[Y + \tau A^{*}(S + BYA^{*})^{-1}B \right].
\end{alignat}

\subsection{Convergence Analysis}
The proof of convergence from the previous section can be generalized to apply to the modified \text{G-AMA} described in Section \ref{ss:gen1} and \ref{ss:gen2}. Since the new iterates belong to a subset of the set ${\cal U}$ defined in Lemma \ref{lowerbound}, the same lemma is valid for the new problems. Lemma \ref{lipschitz}, Lemma \ref{upperbound} also hold for the modified method. In order to prove Lemma \ref{strongconvexlemma}, we first recall the non-expansiveness of the $\mbox{prox}$ operator for any convex set $\bar{\cal D}$ \cite{rockafellar1976monotone}.
\begin{lemma}
\label{proxop}
(\cite{rockafellar1976monotone}): The proximal operator $\mbox{prox}_{\tau\bar{\cal D}}$ for an indicator function of a convex set $\bar{\cal D}$ given by
\begin{alignat}{1}
\mbox{prox}_{\tau\bar{\cal D}}(x)  &=  \argmin{y}\left\{ \bm{1}_{\bar{\cal D}} + \frac{\tau}{2} \|y - x\|^{2}\right\} = \Pi_{\bar{\cal D}}(x)
\end{alignat}
satisfies
\begin{alignat}{1}
\|\mbox{prox}_{\tau\bar{\cal D}}(x) - \mbox{prox}_{\tau\bar{\cal D}}(y)\|_{2} \leq  \|x -y\|_{2}
\end{alignat}
\end{lemma} 
Lemma \ref{proxop} is useful for proving linear convergence for a generalized convex constraint on $\Sigma$. We now formally show global linear convergence of generalized \text{G-AMA}.
\begin{theorem}
\label{t:gengama}
The new iterate $Y_{+}$ of the generalized \text{G-AMA} algorithm from (\ref{boundconupdate}) and (\ref{mod2iter}) satisfy 
\begin{alignat}{1}
\label{mt1}
\|Y_{+} - Y_{*}\|_{F} \leq \gamma \|Y - Y_{*}\|_{F}, 
\end{alignat}
for some $\gamma < 1$ and hence the iterates converges to a $\epsilon$-optimal solution $Y_{*}$ in ${\cal O}\left(\log(1/\epsilon)\right)$ iterations.
\end{theorem}
\begin{proof}
The proof is given in the Appendix.
\end{proof}

\section{Numerical Experiments}
The performance of \text{G-AMA} was compared to other algorithms on synthetic as well as real data sets. We begin with illustrating linear convergence of \text{G-AMA} and how the convergence of \text{G-AMA} varies with condition number of the estimate $\hat{\Sigma}^{*}$. The left hand plot in Fig. \ref{fig:gaplincon} illustrates the convergence of $\|Y_{k}-Y_{*}\|_{F}$ and the right hand plot demonstrates the convergence of the \text{G-AMA} iterates with respect to the duality gap $\Delta_{\mbox{opt}}$. The speed of the algorithm becomes slower as the condition number $\kappa(\Sigma^{*})$ increases.
\begin{figure}[h!]
\begin{center}
\includegraphics[width=\textwidth]{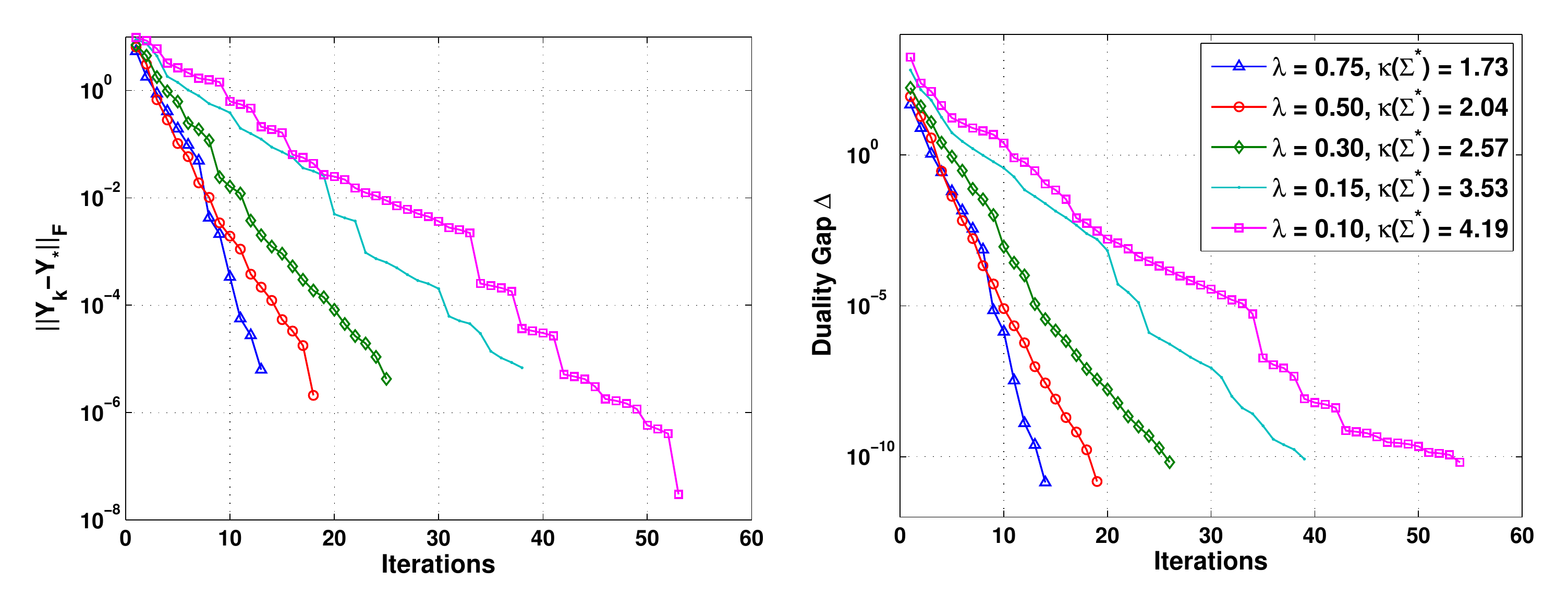}
\end{center}
\caption{Convergence of \text{G-AMA}: (a) $\|Y_{k}-Y_{*}\|_{F}$ and (b) Duality gap $\Delta_{\mbox{opt}}$}
\label{fig:gaplincon}
\end{figure}

Next, we demonstrate timing comparisons between \text{G-AMA}, \text{QUIC} and \text{G-ISTA} for solving Problem (\ref{mainprob}). We test the three algorithms on synthetic as well as real data sets. The timings reported in this section are the wall-clock timings and not the CPU timings. The three algorithms use LAPACK library for computing Cholesky factors and inversion which leverages multiple processors when available. This benefits the first order methods, \text{G-AMA} and \text{G-ISTA} which expend majority of their computational resources in these operations, in contrast to \text{QUIC} which expends computational resources in solving the lasso subproblem.

 
\subsection{Synthetic datasets} 
The first set of comparison experiments were conducted on synthetic data. For this subsection, the data was generated using the procedure in \cite{lu2010adaptive}. For a given problem size $p$, the underlying inverse covariance matrix $\Omega = \Sigma^{-1}$ was generated by choosing the off-diagonal entries i.i.d. from a uniform distribution over $[-1,1]$. To obtain a desired sparsity level of $sp\%$ (percentage of non-zero entries in $\hat{X}$), the off-diagonals were set to zero with probability $sp/100$ (As in \cite{brajarat2012nips}, sparsity levels of $3\%$ and $15\%$ were used). A multiple of the identity matrix was added to this matrix to adjust the smallest eigenvalue of $\Omega = \Sigma^{-1}$ to $1$. This procedure ensures that $\Omega$ is positive definite, sparse, and well-conditioned. Finally, datasets of $n$ i.i.d. samples were generated from the normal distribution ${\cal N}_{p}(0, \Omega^{-1})$. For each $\Omega$, sample sizes of $n=0.2p$ and $n=1.2p$ were chosen to illustrate both when $n << p$ and when $n \geq p$.

 
\subsection*{Heuristics and insights}
\label{hni}
During the numerical experiments it was observed that as $\lambda$ increases, the number of iterations go down. Moreover, for large values of $\lambda$, the performance of each algorithm depends on the initial point that is chosen. While \text{G-ISTA} and \text{QUIC} use $\Omega_{0} = ($\texttt{diag(}$S$\text{)}$+\lambda I)^{-1}$ as the starting point, the inverse $Y_{0,1} = ($\texttt{diag(}$S$\texttt{)}$+\lambda I)$ works very well for \text{G-AMA} for large values of $\lambda$. However, as the value of $\lambda$ decreases, it was observed that the next iterate $Y_{1}$ was not always positive definite. This can be explained by violation of dual feasibility of $Y_{0}$ for small values of $\lambda$. A better consistent starting point which works for all values of $\lambda$ is $Y_{0,2} = S + \lambda I$. The timing comparisons shown in this section display the time required for convergence starting with the best initial point. In the majority of cases, the best initial point was $S + \lambda I$ for \text{G-AMA}. As a general guideline, $Y_{0,2}$ can be used in general, and $Y_{0,1}$ can also be used to obtain some speedup when the expected solution is very sparse ($< 2\%$). We also observe that as $p$ increases \text{G-AMA} and \text{G-ISTA}, scale more favorably than \text{QUIC}. A general observation as $n$ increases is that all the methods perform better. This can be explained by the fact that condition number of the sample covariance $S$ improves with $n$. Therefore, the iterates and the optimal solution have better conditioning.

 
\subsection*{\text{G-AMA} vs \text{G-ISTA} vs \text{QUIC}} 
Similar to \text{G-ISTA} and \text{QUIC}, the \text{G-AMA} was implemented in $C$ with a wrapper for \textit{MATLAB} and the three methods were compared in terms of time taken on a machine running the 32-bit version of Ubuntu 12.10 with an Intel Core i7 870 CPU and 8 GB RAM. Table \ref{tab:compare1} and Table \ref{tab:compare2} display the results for problems with $p = 2000$ and $p=5000$ for varying $n$ and $sp$.
\begin{table}[h!]
\centering
\caption{Time (Iterations) and Duality Gap Comparison for \text{G-AMA} with \text{G-ISTA} and \text{QUIC}}
\label{tab:compare1}
\begin{footnotesize}
\begin{tabular}{|c||c| |c|c| |c|c| |c|c| |c|c|} \hline 
&&\multicolumn{2}{c||}{\textbf{algorithm:}} & \multicolumn{2}{c||}{\text{G-ISTA}} & \multicolumn{2}{c||}{\text{QUIC}} & \multicolumn{2}{c|}{\text{G-AMA}} \\ \hline
\textbf{problem} & $\bm{\lambda}$ & nnz($Z_{*}$) &  $\kappa(Y_{*})$ & \textbf{time/iter} & \textbf{gap} & \textbf{time/iter} & \textbf{gap} & \textbf{time/iter} & \textbf{gap} \\
\hline \hline 
 \multirow{4}{*}{\hspace{-0.08in}$\begin{array}{c} p = 2000; \\ n = 400; \\ sp = 03\% \end{array}$\hspace{-0.08in}} 
& 0.12 &  2.36\% &  2.24 &   17s/27   &5e-10&   28s/9   &8e-10&   \textbf{6s}/15  &7e-11\\
& 0.09 &  7.24\% &  6.92 &   56s/75   &1e-10&   99s/10  &1e-10&  \textbf{27s}/57  &6e-11\\
& 0.06 & 15.15\% & 19.13 &   51s/98   &3e-07&  331s/13  &3e-07&  \textbf{28s}/59  &5e-09\\	
& 0.03 & 27.76\% & 45.67 &  327s/457  &4e-06& 1853s/24  &4e-06&  \textbf{18s}/43  &3e-06\\ 
\hline \hline
 \multirow{4}{*}{\hspace{-0.08in}$\begin{array}{c} p = 2000; \\ n = 2400; \\ sp = 03\% \end{array}$\hspace{-0.08in}} 
& 0.12 &  0.10\% &  1.14 &    9s/18   &1e-12&\textbf{3s}/6&1e-12& \textbf{3s}/6   &1e-12\\
& 0.09 &  0.91\% &  1.48 &   21s/31   &2e-12&    9s/7   &2e-12&   \textbf{6s}/13  &1e-12\\
& 0.06 &  3.06\% &  2.76 &   22s/30   &1e-09&   33s/8   &2e-08&  \textbf{18s}/32  &1e-08\\	
& 0.03 & 14.56\% &  9.87 &   37s/53   &2e-07&  585s/16  &3e-07&  \textbf{25s}/41  &4e-08\\ 
\hline \hline
 \multirow{4}{*}{\hspace{-0.08in}$\begin{array}{c} p = 2000; \\ n = 400; \\ sp = 15\% \end{array}$\hspace{-0.08in}} 
& 0.12 &  2.83\% &  2.80 &    8s/15   &1e-06&   24s/7   &5e-06&   \textbf{7s}/16  &2e-07\\
& 0.09 &  7.80\% & 11.67 &   28s/49   &4e-08&  112s/10  &1e-07&  \textbf{22s}/49  &2e-08\\
& 0.06 & 15.23\% & 28.07 &   90s/150  &8e-07&  387s/13  &1e-06&  \textbf{37s}/63  &4e-07\\	
& 0.03 & 27.37\% & 63.37 &  654s/813  &1e-06& 1990s/24  &1e-06&  \textbf{41s}/67  &4e-08\\ 
\hline \hline
 \multirow{4}{*}{\hspace{-0.08in}$\begin{array}{c} p = 2000; \\ n = 2400; \\ sp = 15\% \end{array}$\hspace{-0.08in}} 
& 0.12 &  0.03\% &  1.10 &    3s/8    &1e-12&    3s/6   &1e-12&   \textbf{2s}/4   &1e-12\\
& 0.09 &  0.63\% &  1.34 &    6s/14   &1e-09&   11s/9   &3e-09&   \textbf{3s}/7   &1e-09\\
& 0.06 &  5.38\% &  4.01 &   26s/32   &9e-08&  129s/12  &1e-07&  \textbf{16s}/31  &5e-09\\	
& 0.03 & 19.92\% & 17.55 &   87s/113  &5e-07&  351s/11  &3e-06&  \textbf{19s}/36  &1e-06\\ 
\hline \hline
\end{tabular}
\end{footnotesize}
\end{table}
\begin{table}[h!]
\centering
\caption{Time (Iterations) and Duality Gap Comparison for \text{G-AMA} with \text{G-ISTA} and \text{QUIC}}
\label{tab:compare2}
\begin{footnotesize}
\begin{tabular}{|c|c| |c|c| |c|c| |c|c| |c|c|} \hline 
&&\multicolumn{2}{c||}{\textbf{algorithm:}} & \multicolumn{2}{c||}{\text{G-ISTA}} & \multicolumn{2}{c||}{\text{QUIC}} & \multicolumn{2}{c|}{\text{G-AMA}} \\ \hline
\textbf{problem} & $\bm{\lambda}$ & nnz($Z_{*}$) &  $\kappa(Y_{*})$ & \textbf{time/iter} & \textbf{gap} & \textbf{time/iter} & \textbf{gap} & \textbf{time/iter} & \textbf{gap} \\
\hline \hline
 \multirow{4}{*}{\hspace{-0.08in}$\begin{array}{c} p = 5000; \\ n = 1000; \\ sp = 3\% \end{array}$\hspace{-0.08in}} 
&0.10 & 0.47\%& 1.43 &  59s/11 &  2e-07 & 102s/7  &  3e-07 & \textbf{41s}/7  & 7e-08 \\
&0.08 & 2.03\%& 2.29 & 215s/28 &  1e-11 & 434s/9  &  1e-11 & \textbf{92s}/16  & 1e-11 \\
&0.06 & 6.36\%& 8.29 & 478s/54 &  9e-09 & 1607s/10  &  9e-09 & \textbf{284s}/47  & 7e-09 \\
&0.04 &13.69\%&23.12 & 935s/127 &  6e-07 & 6315s/14  &  1e-06 & \textbf{322s}/53  & 1e-06 \\
\hline \hline
\multirow{4}{*}{\hspace{-0.08in}$\begin{array}{c} p = 5000; \\ n = 6000; \\ sp = 3\% \end{array}$\hspace{-0.08in}} 
&0.08 & 0.16\%& 1.15 &  76s/11 &  2e-08 &  40s/5  &  9e-08 & \textbf{31s}/5  & 1e-09 \\
&0.06 & 1.12\%& 1.59 &  59s/11 &  2e-07 & 145s/6  &  1e-06 & \textbf{52s}/9  & 9e-09 \\
&0.04 & 3.24\%& 3.65 & 213s/27 &  2e-09 & 563s/8  &  3e-09 & \textbf{167s}/27  & 8e-10 \\
&0.02 &12.89\%&15.03 & 391s/61 &  4e-06 &4682s/12  &  4e-06 & \textbf{247s}/41  & 1e-06 \\
\hline \hline
 \multirow{4}{*}{\hspace{-0.08in}$\begin{array}{c} p = 5000; \\ n = 1000; \\ sp = 15\% \end{array}$\hspace{-0.08in}} 
&0.08 & 2.53\%& 3.31 &\textbf{84s}/12 & 6e-07 &365s/7  &  9e-07 &       135s/20  & 7e-07 \\
&0.06 & 6.94\%&18.20 &  933s/109 & 1e-08 &2099s/11  &  1e-08 & \textbf{507s}/82  & 7e-09 \\
&0.04 &13.88\%&39.37 & 1599s/229 & 4e-05 &5270s/12  &  4e-05 & \textbf{270s}/45  & 4e-05 \\
&0.02 &26.02\%&84.12 & 6865s/1000& 4e-03 &31841s/24  &  3e-06 & \textbf{390s}/62  & 2e-07 \\
 \hline \hline
\multirow{4}{*}{\hspace{-0.08in}$\begin{array}{c} p = 5000; \\ n = 6000; \\ sp = 15\% \end{array}$\hspace{-0.08in}} 
&0.08 & 0.04\% & 1.10 & 66s/11  & 1e-10 & 31s/5  & 3e-09 & \textbf{24s}/4  & 4e-11 \\
&0.06 & 0.72\% & 1.43 & 53s/10  & 8e-07 & 79s/5  & 1e-05 & \textbf{35s}/6  & 8e-08 \\
&0.04 & 5.50\% & 7.14 &236s/31  & 1e-07 &1386s/10& 8e-07 & \textbf{218s}/37 & 9e-07 \\
&0.02 &18.66\% &26.96 &1771s/218& 4e-07 &7132s/12& 4e-06 & \textbf{225s}/39 & 3e-07 \\
 \hline \hline
\end{tabular}
\end{footnotesize}
\end{table}

The first order methods \text{G-AMA} and \text{G-ISTA} are terminated using a tolerance condition on the duality gap. However, the termination of \text{QUIC} is controlled by the tolerance on sub-gradient condition only. The problems were solved using \text{QUIC} by setting the tolerance $\epsilon_{\text{tol}} = 10^{-10}$ and the corresponding duality gap achieved by \text{QUIC} was then used for \text{G-AMA}, and \text{G-ISTA}. The times reported for \text{G-AMA} and \text{G-ISTA} in Tables \ref{tab:compare1} and \ref{tab:compare2} are the times required to achieve a better duality gap than \text{QUIC}. This explains the different duality gaps reported for each of the three methods.

Note that the time required for convergence using \text{G-AMA} is substantially lower than \text{G-ISTA} and \text{QUIC} in all instances except one. Moreover, the advantage of \text{G-AMA} is even more prominent when $\lambda$ is small and/or when the solution is ill-conditioned. As discussed in the next Subsection, this benefit of \text{G-AMA} is magnified when working with real data where the optimal solutions are often highly ill-conditioned.

 
\subsection{Real dataset: \textit{Estrogen} and \textit{Temperature}} 
We compare the three methods on two real datasets \textit{Estrogen} from \cite{estrogen} and \textit{Temperature} from \cite{temperature}. The \textit{Estrogen} dataset consists of expression data for $p=682$ genes from $n=158$ patients with breast cancer. The \textit{Temperature} dataset consists of average annual temperature measurements from $p=1732$ locations over $n=157$ years (1850-2006). The values of $\lambda$ chosen to test the algorithms varied in order to get sparsity between $2\%-12\%$. Both datasets yield extremely ill-conditioned problems for the above values of $\lambda$. We know that the ability of \text{G-AMA} to control the duality gap is extremely useful for attaining accurate solutions in such highly unstable regime. A maximum of $5000$ iterations were used for \text{G-ISTA} and \text{G-AMA}. 
\begin{table}[h!]
\centering
\caption{Wall times, iterations and duality gap comparison for \text{G-AMA} with \text{G-ISTA} and \text{QUIC}. 
 \textit{(Note that the reported duality gap for \text{QUIC} cannot be reduced further for any tolerance set for the solver. For timing comparison purposes, \text{G-AMA} therefore reports two times. The first achieves a duality gap at least as small as \text{QUIC}. The last column$^*$ reports the time to achieve a significantly lower fixed duality gap of $\epsilon_{\texttt{opt}} = 10^{-10}$, a gap that \text{QUIC} cannot achieve for any subgradient tolerance. \text{G-ISTA} is significantly slower or fails to converge within 5000 iterations for these cases.)}}
\label{tab:compare3}
\begin{footnotesize}
\begin{tabular}{|c|c|c| |c|c| |c|c| |c|c| |c|} \hline 
&\multicolumn{2}{c||}{\textbf{algorithm:}} & \multicolumn{2}{c||}{\text{G-ISTA}} & \multicolumn{2}{c||}{\text{QUIC}} & \multicolumn{3}{c|}{\text{G-AMA}} \\ \hline
$\bm{\lambda}$ & nnz($Z_{*}$) &  $\kappa(Y_{*})$ & \textbf{time/iter} & \textbf{gap} & \textbf{time/iter} & \textbf{gap} & \textbf{time/iter} & \textbf{gap} & \textbf{time}\\ \hline 
 \multicolumn{10}{|c|}{\textit{Dataset: Estrogen} p = 682; n = 158} \\ \hline
\hline
0.40 &    2.62\% &  42.36 &   21s/532  & 1e-06 &   \textbf{3s}/13 & 2e-06 &        {5s/180}    & 2e-06 &   \textbf{10s}* \\ \hline
0.30 &    3.39\% &  88.57 &   35s/911  & 7e-07 &   \textbf{6s}/19 & 2e-07 &        {17s/548}   & 4e-08 &   \textbf{21s}* \\ \hline
0.20 &    4.39\% & 192.70 &   80s/2182 & 7e-07 &         23s/27   & 8e-07 &   \textbf{16s}/488 & 8e-07 &   \textbf{27s}* \\ \hline
0.10 &    7.10\% & 475.24 &   NA/5000  &   NA  &         64s/49   & 2e-06 &   \textbf{11s}/339 & 2e-06 &   \textbf{22s}* \\ \hline
0.05 &   12.36\% & 985.99 &   NA/5000  &   NA  &        335s/93   & 2e-06 &   \textbf{11s}/321 & 1e-06 &   \textbf{16s}* \\ \hline 
\hline
 \multicolumn{10}{|c|}{\textit{Dataset: Temperature} p = 1732; n = 157} \\ \hline
0.30 &   1.75\% &  589.70 &  1242s/2940 & 5e-07 & \textbf{279s}/21& 1e-06 &        {438s}/1161 & 9e-07 &  \textbf{639s}* \\ \hline    
0.20 &   2.02\% & 1076.07 &    NA/5000 &   NA  &     589s/28      & 4e-06 & \textbf{356s}/931  & 3e-06 &  \textbf{642s}* \\ \hline     
0.10 &   3.30\% & 2301.98 &    NA/5000 &   NA  &     953s/47      & 2e-05 & \textbf{216s}/581  & 1e-05 &  \textbf{301s}* \\ \hline
0.05 &   5.59\% & 4505.01 &    NA/5000 &   NA  &    3402s/90      & 8e-06 & \textbf{249s}/652  & 6e-06 &  \textbf{374s}* \\ \hline
0.03 &   7.84\% & 7201.09 &    NA/5000 &   NA  &    9433s/136     & 2e-05 & \textbf{209s}/580  & 9e-06 &  \textbf{346s}* \\ \hline
\hline
\end{tabular}
\end{footnotesize}
\end{table}

It was observed that the duality gap achieved by \text{QUIC} decreases with tolerance on the subgradient. However, it stalls at a particular point despite decreasing the tolerance provided to the solver. The duality gaps reported in Table \ref{tab:compare3} for \text{QUIC} are the best that could be achieved by any tolerance on subgradient condition. These are as high as $2\times10^{-6}-2\times10^{-5}$ when the problem is highly ill-conditioned. In contrast, there are two timings reported for \text{G-AMA}. Similar to synthetic experiments, the first is the time required to achieve a duality gap at least as small as the gap achieved by \text{QUIC}. The second time (last column$^{*}$) is that required to achieve a fixed duality gap of $\epsilon = 10^{-10}$. We note that \text{QUIC} is unable to achieve this level of duality gap for any subgradient tolerance, whereas \text{G-ISTA} is very slow in such problems. \text{G-AMA} suffers from neither of these two issues.

Now, as seen in Table \ref{tab:compare3}, for both examples the condition number of the optimal solution is extremely high for lower values of $\lambda = 0.20, 0.10, 0.05, 0.03$. The \text{QUIC} algorithm performs better than \text{G-AMA} when $\lambda$ is higher, but is extremely slower in comparison as $\lambda$ reduces. For the \textit{Estrogen} dataset, \text{G-AMA} is up to 30 times faster than \text{QUIC} when the condition number of solution is almost 1000. For the \textit{Temperature} dataset, when the sparsity level is around $8\%$ for $\lambda = 0.03$, \text{G-AMA} is 45 times faster than \text{QUIC}. Such speedup is extremely vital for applications where regularized maximum likelihood estimation problem need to be solved several times. In particular, determining the level of the penalty parameters is often done through cross-validation which in turn requires solving the $\ell_1$ regularized maximum-likelihood estimation problem for a grid of values of the penalty parameter [see for example \cite{schneider2001analysis} for the climate field reconstruction application]. In addition, when uncertainty quantification is required, both parametric and non-parametric bootstrap require solutions of the optimization problem numerous times. 

\subsection{Portfolio optimization using sparse inverse covariance estimates}
We now apply \text{G-AMA} for solving the regularized maximum likelihood covariance estimation problem in the context of financial portfolio optimization.  A regularized inverse covariance estimate is a critical ingredient in determining the minimum variance portfolio. Moreover, extremely fast estimation methods like \text{G-AMA} are required for such problems as covariance matrices need to be estimated repeatedly over moving time blocks. In this section we shall compare the performance of \text{G-AMA}with popular methods in the literature to illustrate its efficacy in the financial portfolio optimization context.

\subsection*{Minimum variance portfolio with rebalancing}
A portfolio optimization problem refers to the problem of determining weights or proportions (in monetary terms) in order to invest in a set of securities that minimize the risk for a given level of return. We will focus on the minimum variance portfolio using the covariance estimate computed by solving problem \ref{mainprob}. For a portfolio with $p$ risky assets, let $r_{i}$ denote the return of asset $i$ over a given period, \ie, its change in price over one time period divided by its price at the beginning of the period. Let $\Sigma$ denote the covariance matrix of $r = (r_{1}, r_{2}, \ldots, r_{p})$ and $w_{i}$ be the weight of asset $i$ in the portfolio during a given period ($w_{i}$ could be positive or negative based on long or short positions). The minimum variance portfolio selection problem solves
\begin{alignat}{1}
\minimize{w} \ &  w^{T} \Sigma w \nonumber \\
\st \ & \mathbf{1}^{T}{w} = 1,
\end{alignat} 
where the objective denotes the variance of the return (also defined as the risk) associated with a particular portfolio. The linear constraint represents the budget constraint. For a given $\Sigma$, this quadratic program has a closed form analytic solution $w^{*} = (\mathbf{1}^{T}\Sigma^{-1}\mathbf{1})^{-1}\Sigma^{-1}\mathbf{1}$.

The standard problem defined above assumes stationarity of the returns. In order to account for non-stationarity, we employ the minimum variance portfolio rebalancing (MVR) strategy. This approach updates the portfolio weights every $L$ units of time, effectively dividing the trading horizon into blocks each consisting of $L$ time units. At the start of each block, the minimum variance portfolio problem is solved based on the past estimation horizon size of $N_{\text{estim}}$ observations of returns. The weights $w$ are then held constant for $L$ time units during the holding periods. We assume that the total number of time units in the entire period, $N_{\text{tot}} = N_{\text{estim}} + K \cdot L$, for some positive integer $K$. Therefore, we have $K$ updates of portfolio weights, say $w^{(j)}$ over the holding period of $\left[N_{\text{estim}} + 1 + (j-1) \cdot L, N_{\text{estim}} + j \cdot L\right]$ for $j = 1,\ldots K$. Here, $w^{(j)}$ is calculated using the estimated covariance of asset returns over the $j$th holding period.

\subsection*{Application to the DJIA}
In the numerical study, we use the stock data for 29 stocks that constituted the Dow Jones Industrial Average (\text{DJIA}) as of July 2008, and had a start date of January 01, 1995 and an end date of October 26, 2012 (4473 trading days). We have chosen this basket of securities deliberately for comparison purposes as covariance estimation methods have been recently illustrated on this data. The first rebalancing interval begins from January 01, 1995 and the last interval begins on July 02, 2012. We have,
\[
N_{\text{tot}} = 4473, \hspace{0.2in} K = 58, \hspace{0.2in} L = 80, \hspace{0.2in} N_{\text{estim}} = 75, 150, 225, 300.
\] 
We estimate the sample covariance $S$ using $N_{\text{estim}}$ weekly returns from the \text{DJIA}. The entire trading horizon consists of $K = 58$ holding periods of $L = 80$ days (16 weeks) each. We compare the MVR strategy using covariance matrices generated by \text{SparseConc} using \text{G-AMA} (with constant penalty parameter $\lambda = 0.1\|S\|_2$) versus covariance estimation methods considered in \cite{won2012condition}. These include \textbf{CondReg} (condition number regularized covariance from \cite{sang2012condreg}), \textbf{Factor} (a shrinkage method from \cite{ledoit2003honey}), \textbf{LedoitWolf} (a linear shrinkage scheme from \cite{ledoit2004well}), \textbf{Sample} ( sample covariance matrix) as well as the \textbf{DJIA}. The performance of these methods is compared on the metrics of \textit{Realized Return}, \textit{Realized Risk}, \textit{Realized Sharpe Ratio} (ratio of excess expected return and risk), \textit{Turnover} (measure of portfolio switching for each time-interval), \textit{Normalized Wealth Growth} and the \textit{Size of the Short Side} (measure of negative weights chosen). The formal definitions for the performance metrics are given in Appendix C.

\subsection*{Performance metrics}
Fig. \ref{fig:expt9} shows the normalized wealth growth over the trading horizon for five methods and the \text{DJIA} for four different values of $N_{\text{estim}}$. As can be seen, the wealth growth using G-AMA estimates give the best results compared to other methods. Table \ref{rrrrsr} shows the values for realized return, realized risk, Sharpe ratio, turnover and size of short side. Each entry is the mean (standard deviation) of the corresponding metric over the trading period. The standard deviations are computed in a heteroskedasticity-and-autocorrelation consistent manner discussed in \cite[Sec. 3.1]{ledoit2008robust}. 
The Turnover values are illustrated in Fig. \ref{fig:expt9b}.
\begin{figure}[htp!]
\begin{center}$
\begin{array}{cc}
\includegraphics[width=2.75in]{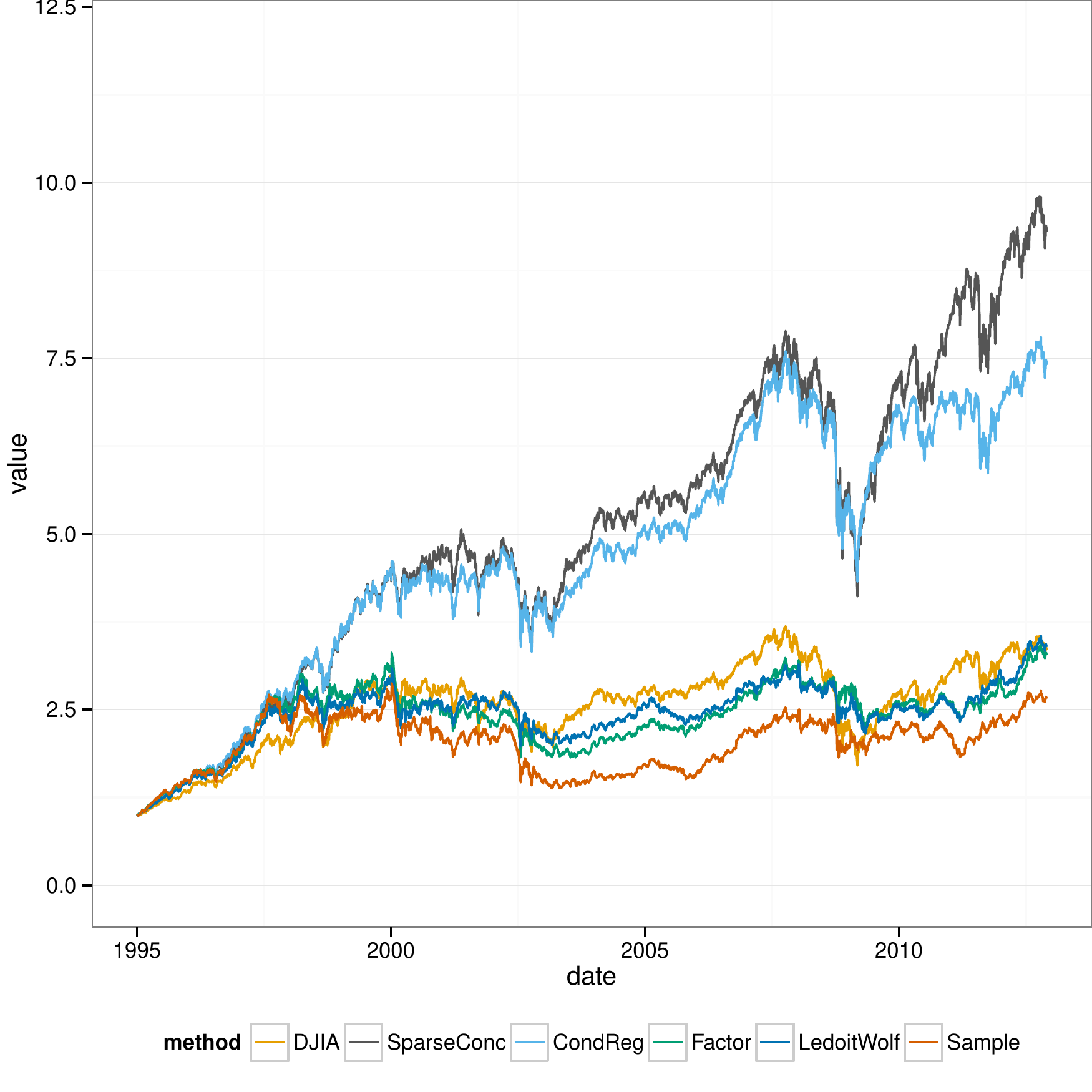} &
\includegraphics[width=2.75in]{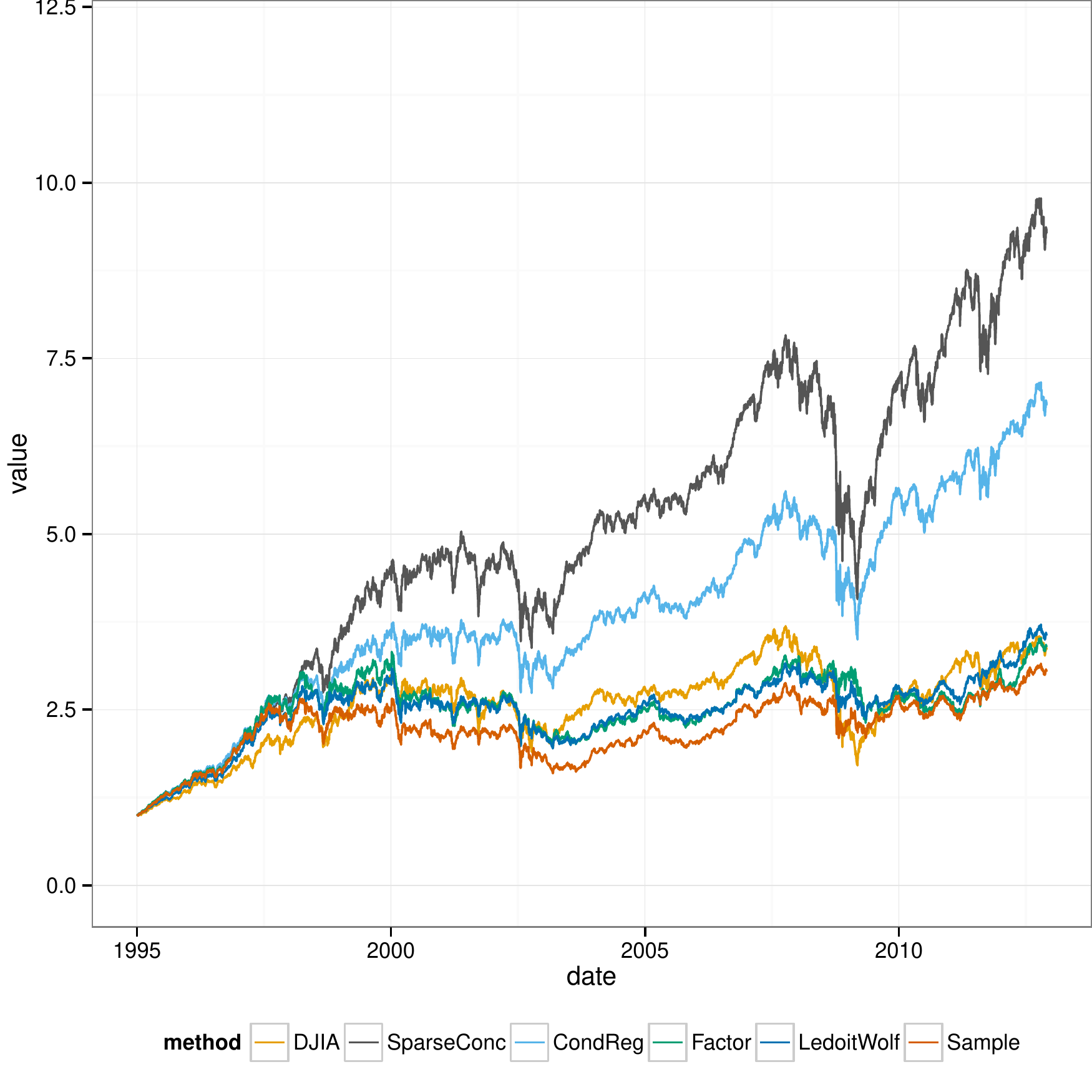} \\
\includegraphics[width=2.75in]{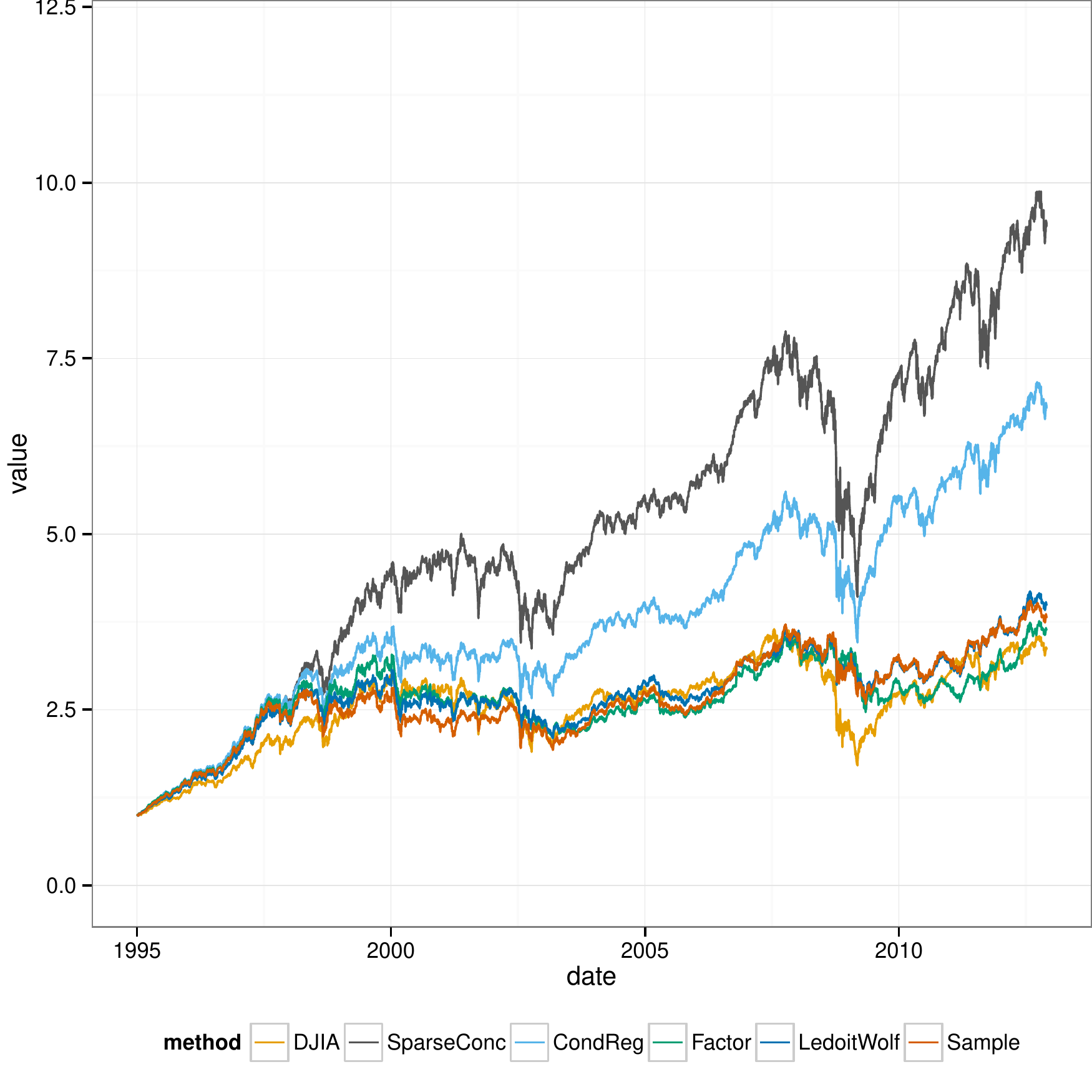} &
\includegraphics[width=2.75in]{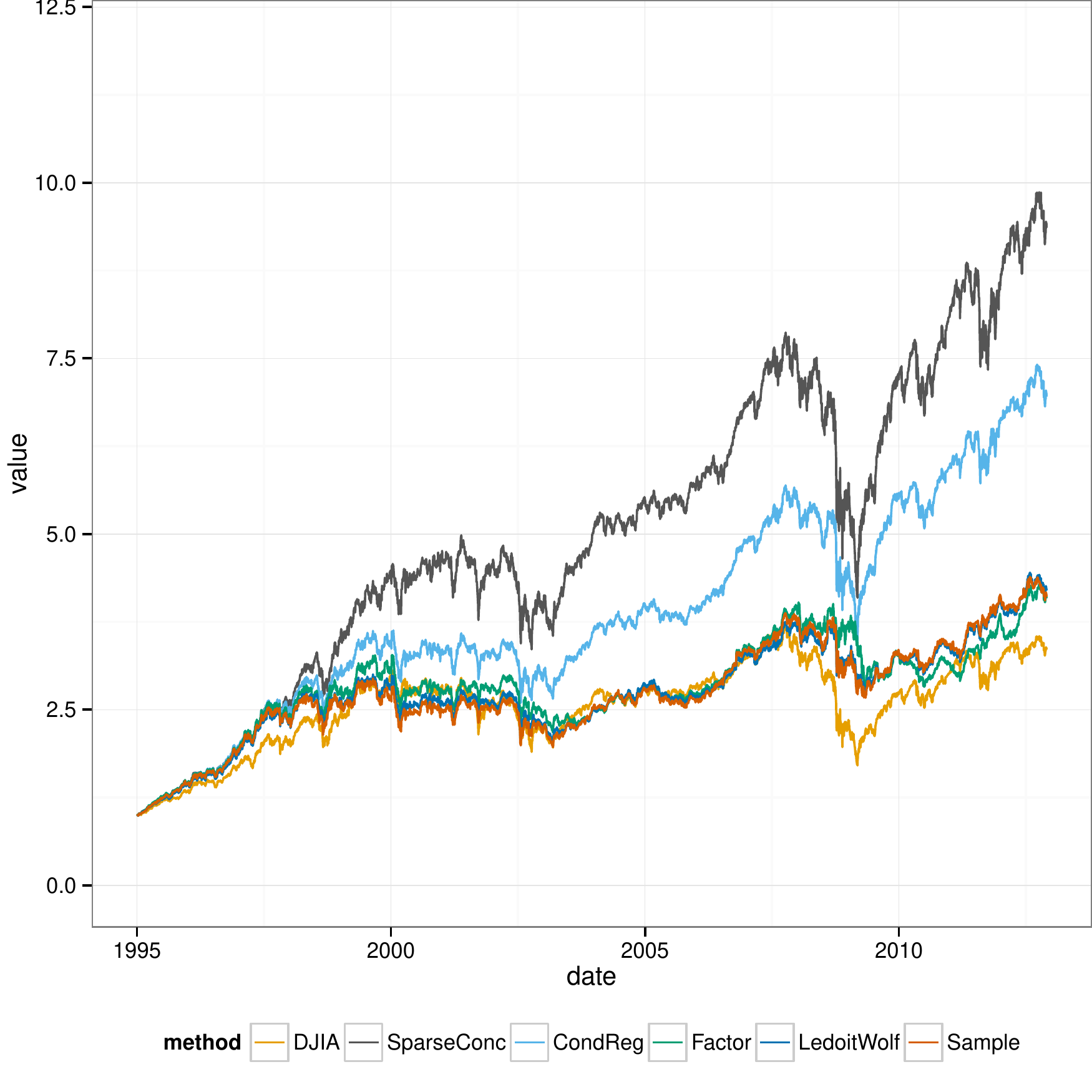} 
\end{array}$
\end{center}
\caption{Normalized wealth growth of \text{G-AMA} vs other methods ($\lambda = 0.10$)}
\label{fig:expt9}
\end{figure}
\begin{table}
\centering
\caption{Realized Return, Realized Risk, Sharpe Ratio, Turnover and Size of Short Side}
\label{rrrrsr}
\begin{footnotesize}
\begin{tabular}{|c||c|c|c|c|c|}
\hline
$\begin{array}{c}\mbox{covariance} \\ \mbox{regularization scheme}\end{array}$
& $\begin{array}{c}\mbox{return [\%]}\end{array}$
& $\begin{array}{c}\mbox{risk [\%]}\end{array}$
& $\begin{array}{c}\mbox{Sharpe ratio}\end{array}$
& $\begin{array}{c}\mbox{Turnover}\end{array}$
& $\begin{array}{c}\mbox{Size of short side}\end{array}$
\\ \hline \hline
& \multicolumn{5}{c|}{$N_{\text{estim}}=75$}\\ \hline
\text{SparseConc}&  14.376 (4.12) & 19.035 (0.60) & \textbf{0.493} (0.22) & \textbf{0.186} (0.04) &  \textbf{0.000} (0.00)\\
\text{CondReg}   &  13.256 (3.83) & 17.681 (0.54) & 0.467 (0.22) & 0.593 (0.44) &  0.043 (0.07) \\
\text{Factor}    &   9.581 (3.58) & 15.722 (0.44) & 0.291 (0.23) & 1.809 (0.45) &  0.219 (0.05) \\ 
\text{LedoitWolf}&   9.687 (3.61) & 15.350 (0.40) & 0.305 (0.24) & 1.785 (0.54) &  0.201 (0.07) \\
\text{Sample}    &   9.545 (4.12) & 17.486 (0.34) & 0.260 (0.24) & 2.673 (0.76) &  0.303 (0.06) \\
\hline
& \multicolumn{5}{c|}{$N_{\text{estim}}=150$}\\ \hline
\text{SparseConc}&  14.353 (4.12) & 19.043 (0.60) & \textbf{0.491} (0.22) & \textbf{0.177} (0.04) &  \textbf{0.000} (0.00) \\
\text{CondReg}   &  12.917 (3.61) & 16.493 (0.50) & 0.480 (0.22) & 0.921 (0.37) & 0.090 (0.06) \\
\text{Factor}    &   9.690 (3.51) & 15.397 (0.42) & 0.305 (0.23) & 1.781 (0.45) & 0.213 (0.06) \\
\text{LedoitWolf}&   9.837 (3.41) & 14.868 (0.41) & 0.325 (0.23) & 1.748 (0.53) & 0.194 (0.07) \\
\text{Sample}    &   9.487 (3.59) & 15.728 (0.44) & 0.285 (0.23) & 2.144 (0.68) & 0.254 (0.06) \\
\hline
& \multicolumn{5}{c|}{$N_{\text{estim}}=225$}\\ \hline
\text{SparseConc}&  14.407 (4.12) & 19.053 (0.60) & \textbf{0.494} (0.22) & \textbf{0.173} (0.04) &  \textbf{0.000} (0.00)\\
\text{CondReg}   &  12.934 (3.55) & 16.131 (0.49) & 0.492 (0.22) & 1.035 (0.20) & 0.102 (0.04) \\
\text{Factor}    &  10.011 (3.47) & 15.256 (0.42) & 0.328 (0.23) & 1.751 (0.48) & 0.208 (0.06) \\
\text{LedoitWolf}&  10.411 (3.34 )& 14.690 (0.41) & 0.368 (0.23) & 1.719 (0.54) & 0.190 (0.07) \\
\text{Sample}    &  10.482 (3.44) & 15.144 (0.41) & 0.362 (0.23) & 1.980 (0.65) & 0.235 (0.07) \\
\hline
& \multicolumn{5}{c|}{$N_{\text{estim}}=225$}\\ \hline
\text{SparseConc}&  14.403 (4.13) & 19.085 (0.60) & 0.493 (0.22) & \textbf{0.170} (0.03) &  \textbf{0.000} (0.00)\\
\text{CondReg}   &  13.077 (3.54) & 16.164 (0.49) & \textbf{0.499} (0.22) & 1.037 (0.16) & 0.098 (0.03)\\
\text{Factor}    &  10.656 (3.47) & 15.262 (0.41) & 0.371 (0.23) & 1.735 (0.48) & 0.205 (0.06)\\ 
\text{LedoitWolf}&  10.680 (3.33) & 14.660 (0.40) & 0.387 (0.23) & 1.699 (0.54) & 0.186 (0.07)\\
\text{Sample}    &  10.809 (3.43) & 15.076 (0.40) & 0.385 (0.23) & 1.890 (0.62) & 0.223 (0.07) \\
\hline
\hline
\end{tabular}
\end{footnotesize}
\end{table}
\begin{figure}[htp!]
\begin{center}$
\begin{array}{cc}
\includegraphics[height=2.5in, width=2.75in]{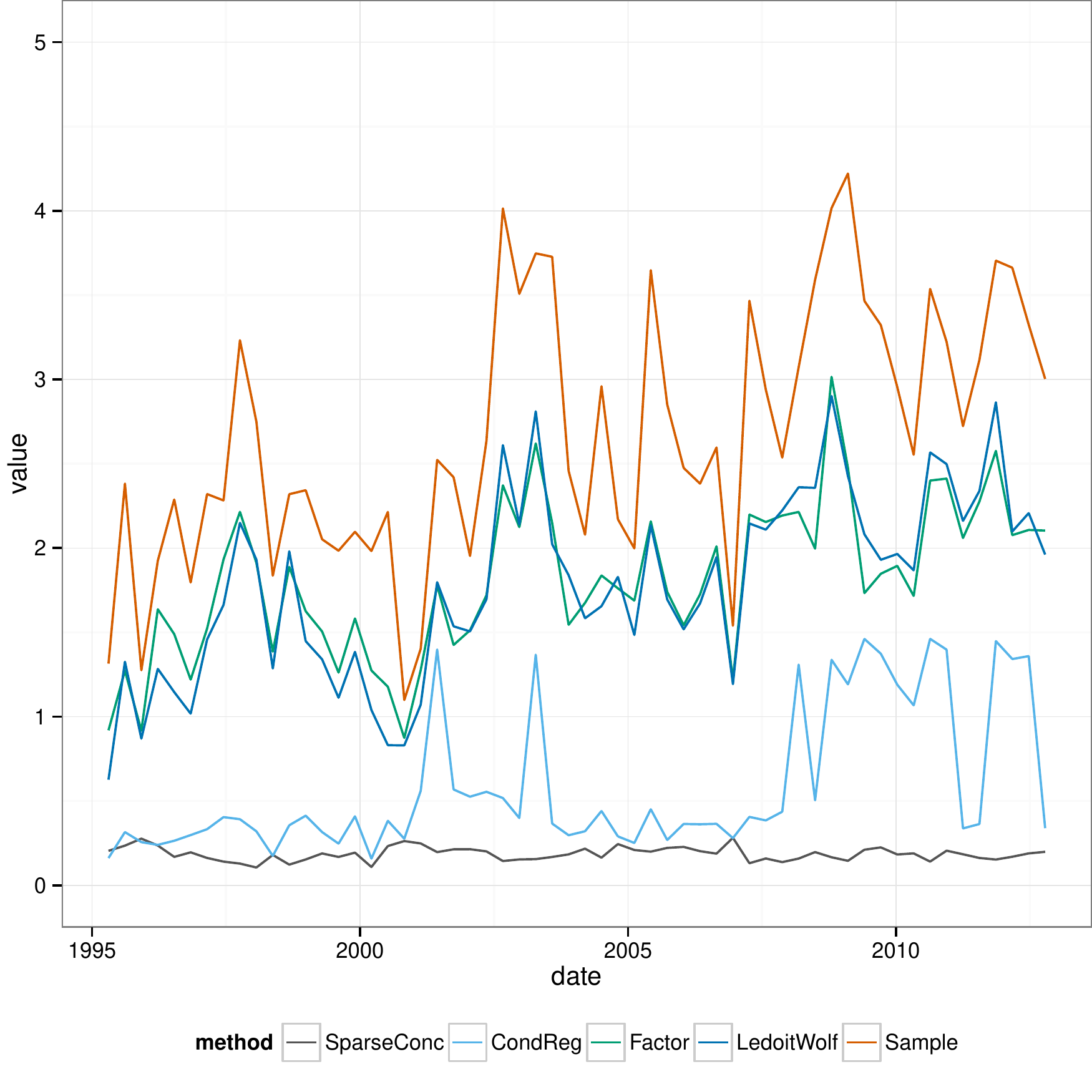} &
\includegraphics[height=2.5in, width=2.75in]{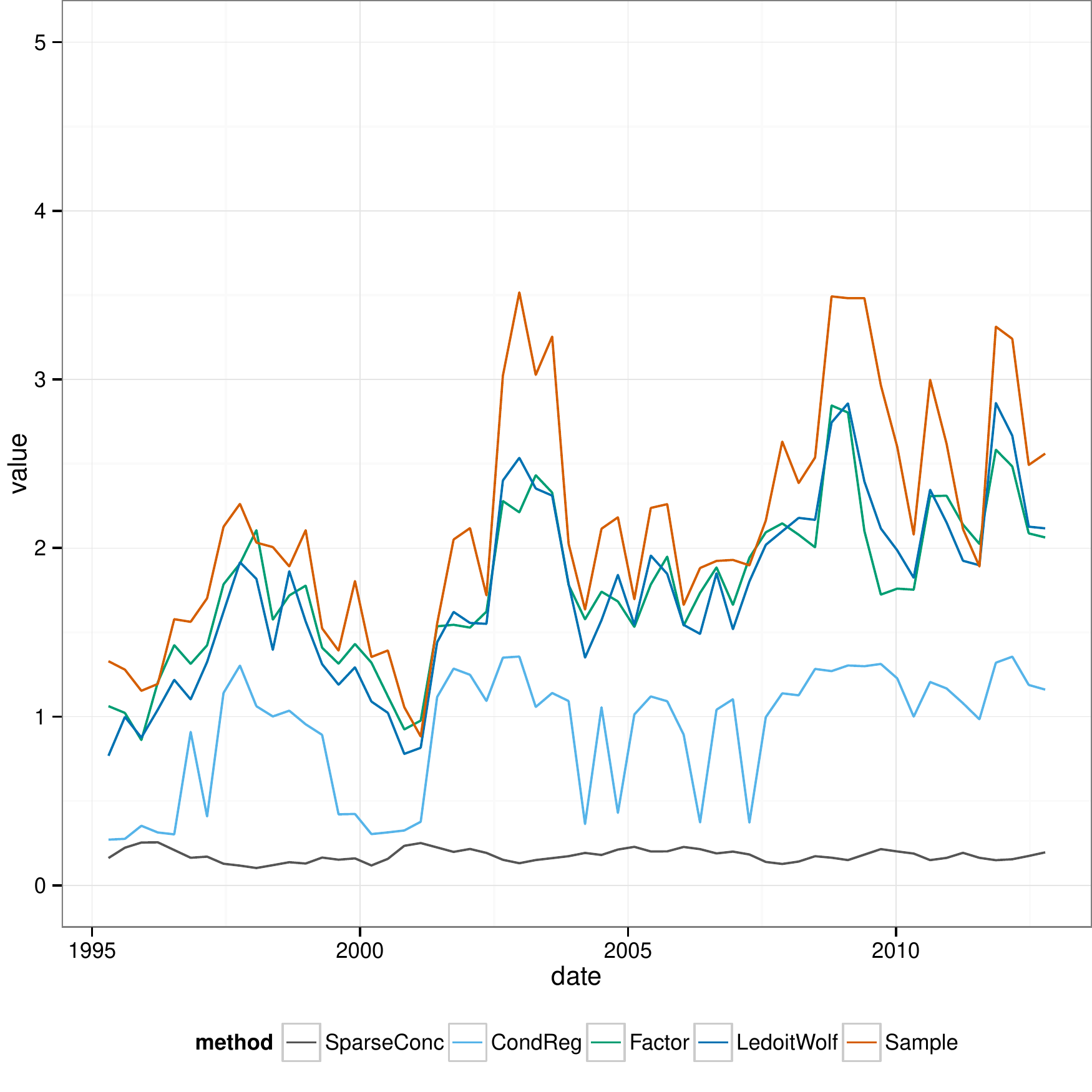} \\
\includegraphics[height=2.5in, width=2.75in]{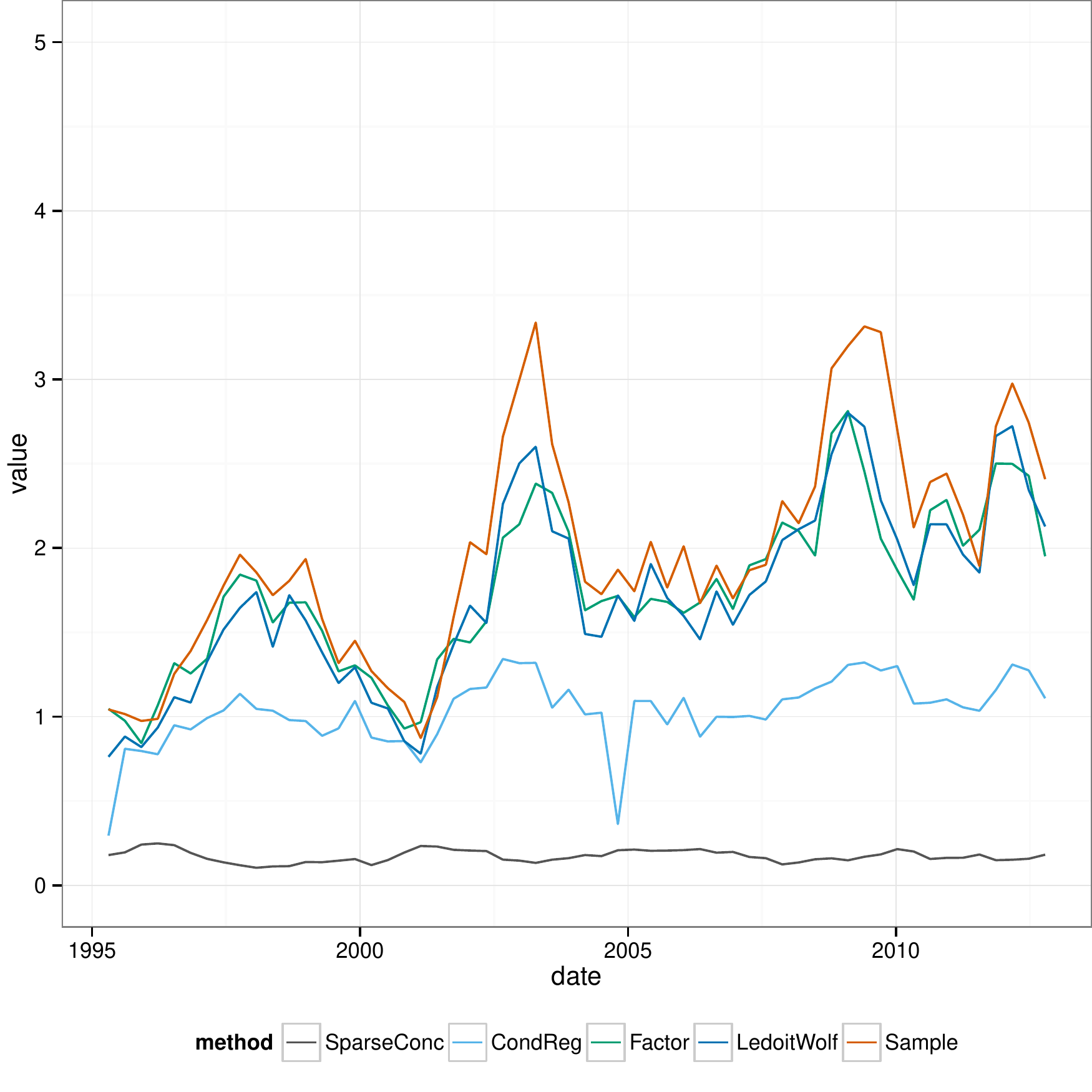} &
\includegraphics[height=2.5in, width=2.75in]{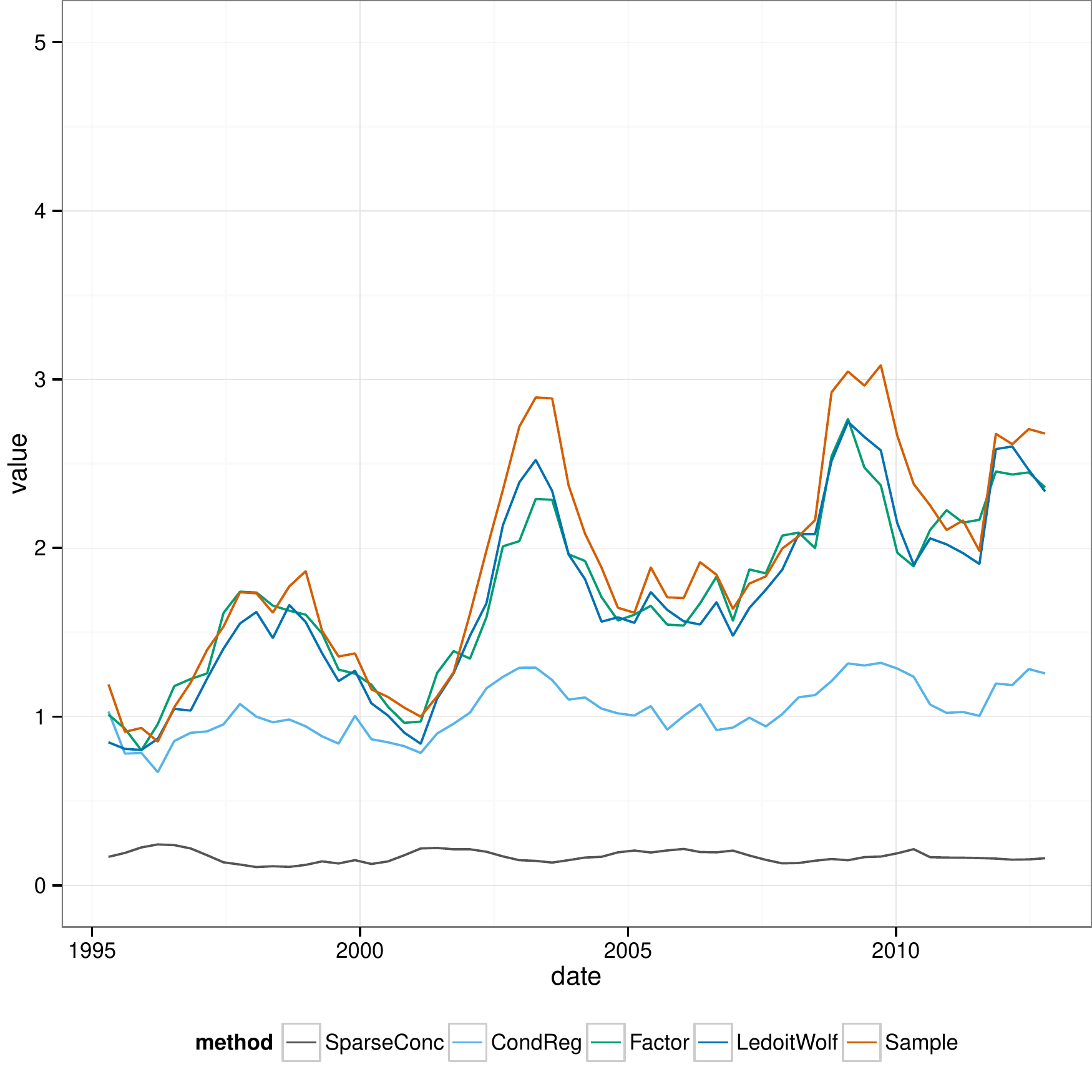} 
\end{array}$
\end{center}
\caption{Turnover of \text{G-AMA} vs. other methods ($\lambda = 0.1\times\|S\|_{2}$)}
\label{fig:expt9b}
\end{figure}

\subsection*{Comparison Results}
The normalized wealth growth for \text{G-AMA} is substantially better than other methods. G-AMA also gives high values of Sharpe ratio across estimation horizons. An advantage of \text{G-AMA} is that the turnover is lower than for other methods and the size of short-side is negligible. G-AMA gives rise to a stable portfolio which avoids excessive continuous re-balancing, transaction and borrowing costs. So in this sense, G-AMA's net effective wealth growth are even higher than those depicted in Fig. \ref{fig:expt9}.

\section{Conclusion}
This paper proposes a novel sparse inverse covariance estimation method \text{G-AMA} that adds and extends the sparse inverse covariance estimation literature in different directions. In particular, \text{G-AMA} is numerically faster than the previous best algorithms by several orders of magnitude in several regimes, while retaining all the good theoretical convergence properties of recently proposed methods. Moreover, \text{G-AMA} can solve ill-conditioned problems more accurately where existing methods do not converge or are extremely slow. The extra speed and ability to solve poorly conditioned problems makes \text{G-AMA} highly attractive for modern day applications, where the inverse covariance matrix has to be estimated multiple times. The \text{G-AMA} formulation also facilitates incorporating domain specific information while still maintaining linear convergence.

\appendix
 
\section{Algorithm Details}

\subsection{Termination criteria}
The algorithm is terminated when $Z_{k}$ is positive definite and a tolerance of $\epsilon_{\texttt{opt}}$ is attained by the duality gap. The duality gap is given by
\begin{alignat}{1}
\Delta_{\text{opt}} &= -\log\det Y_{k} - p - \log\det Z_{k} + \langle S, Z_{k} \rangle + \lambda\|Z_{k}\|_{1}. \nonumber
\end{alignat}
When the algorithm terminates, the covariance estimate and the sparse inverse covariance estimate are given by $Y_{k+1}$ and $Z_{k+1}$ respectively.

An additional criteria of tolerance of $\epsilon_{\text{prim}}$ can be imposed on the progress of $Y_{k}$ iterates
\begin{alignat}{1}
\Delta_{\text{prim}} &= \frac{\|Y_{k+1}-Y_{k}\|_{F}}{\|Y_{k}\|_{F}}. \nonumber
\end{alignat}
This quantity is related to the primal constraint violation as
\begin{alignat}{1}
\text{primal feasibility} &= \|X_{k}-Z_{k}\|_{F} = \frac{\|Y_{k+1}-Y_{k}\|_{F}}{\tau}. \nonumber
\end{alignat}
Therefore, the tolerance of $\epsilon_{\text{prim}}$ indirectly imposes primal feasibility. 

\subsection{Choice of step size}
Recall that the step size has to be chosen in equation (5). The step size for each iteration is chosen starting with an adaptive Barzilia-Borwein (\text{BB}) step \cite{barzilai1988two}. 

For a gradient descent algorithm minimizing a convex function $f(x)$ with gradient $\nabla{f}(x) = g$, the BB step size $\tau$ maximizes 
\begin{alignat}{1}
\label{bbss}
\tau &= {\argmax{\tau} \|\tau \Delta x - \Delta g\|_{2}} = \frac{\la \Delta x, \Delta x\ra}{\la \Delta g, \Delta x\ra},
\end{alignat}
which is a two-point approximation to the secant equation. 

An equivalent step size for \text{G-AMA} solving \ref{bbss} gives
\begin{alignat}{1}
\label{stepsize}
\tau_{k+1} &= \frac{\la Y_{k+1}-Y_{k}, Y_{k+1}-Y_{k} \ra}{\la Y_{k+1}-Y_{k}, X_{k}-X_{k+1} \ra}
\end{alignat}
is used in the numerical experiments. This heuristic provides substantial acceleration for convergence as compared to the fixed step size. However, the \text{BB} step might not always be feasible or satisfy sufficient descent condition. Hence, starting with the \text{BB} step we backtrack until the sufficient descent condition is satisfied.

\section{Lemmas and Proofs from Section 3}

\subsection{Proof of Lemma \ref{fixedpointlemma}}
\begin{proof}
Starting with $X_{*}$ satisfying (\ref{optcond}), we have three cases. \\
If $X_{*}(i,j) = 0$, 
$$X_{*}^{-1}(i,j) - S(i,j) = {\cal C}_{\lambda}\left(\tau X_{*}(i,j) + (X_{*}^{-1}(i,j) - S(i,j))\right), \ \ \because \ \ \tau X_{*}(i,j) = 0.$$
Else if $X_{*}(i,j) > 0$ and $X_{*}^{-1}(i,j) - S(i,j) = \lambda$, we have 
$${\cal C}_{\lambda}\left(\tau X_{*}(i,j) + (X_{*}^{-1}(i,j) - S(i,j))\right) = \lambda, \ \ \because \ \  \tau X_{*}(i,j) > 0.$$
And if $X_{*}(i,j) < 0$ and $X_{*}^{-1}(i,j) - S(i,j) = -\lambda$, we have 
$${\cal C}_{\lambda}\left(\tau X_{*}(i,j) + (X_{*}^{-1}(i,j) - S(i,j))\right) = -\lambda, \ \ \because \ \ \tau X_{*}(i,j) < 0.$$
Hence, the inverse of $X_{*}$ is a fixed point of (\ref{fixedpoint}). 

For the reverse result, we begin with the inverse of a fixed point $X_{*} = Y_{*}^{-1}$. Therefore, we have,
\begin{alignat}{1}
X_{*}^{-1} - S &= {\cal C}_{\lambda}\left(\tau X_{*} + (X_{*}^{-1} - S) \right) \nonumber
\end{alignat}
The entries on the right hand side are contained in the interval $[-\lambda, \lambda]$ and the additional term of $\tau X_{*}$ in the argument of clip function implies that $X_{*}(i,j)$ has to be $0$ when $|X_{*}^{-1} - S| < \lambda$ and $X_{*}(i,j)$ is non-zero only if $(X_{*}^{-1} - S)(i,j) = \pm \lambda$ with same sign as $X_{*}(i,j)$. Thus $X_{*}^{-1} - S$ satisfy both the conditions in (\ref{optcond}) proving the required result.
\end{proof}

\subsection{Contraction of \text{G-AMA} iterates}
Using Lemma \ref{lipschitz}, we now prove the following lemma which is a key ingredient for showing upper bound on iterates of \text{G-AMA} as well as global linear convergence of \text{G-AMA}.
\begin{lemma}
\label{strongconvexlemma}
Let 
\begin{alignat}{1}
Y_{+} = {\cal C}_{\lambda}\left( Y + \tau Y^{-1} - S \right) + S, 
\end{alignat}
be from (\ref{gamaiter}) and let $Y_{*}$ be a fixed point from equation (\ref{fixedpoint}), i.e., 
\begin{alignat}{1}
\label{scl2}
Y_{*} = {\cal C}_{\lambda}\left( Y_{*} + \tau Y_{*}^{-1} - S \right) + S. 
\end{alignat}
Also, define
\begin{alignat}{1}
\alpha = \min\left\{ \lambda_{\min}(Y), \lambda_{\min}(Y_{*}) \right\}, \beta = \max\left\{ \lambda_{\max}(Y), \lambda_{\max}(Y_{*}) \right\}. \nonumber
\end{alignat}
Then  
\begin{alignat}{1}
\|Y_{+} - Y_{*}\|_{F} \leq \displaystyle\max\left\{\bigg|1-\frac{\tau}{\alpha^{2}}\bigg|, \bigg|1-\frac{\tau}{\beta^{2}}\bigg|\right\} \|Y - Y_{*}\|_{F}. \nonumber
\end{alignat}
\end{lemma}
\begin{proof}
We begin with the left hand side and substituting (\ref{gamaiter}) and (\ref{fixedpoint}) and using the non-expansive property (see Lemma \ref{proxop}) of the ${\cal C}_{\lambda}$ function (projection onto a convex set), we get
\begin{alignat}{1}
\label{above2}
\|Y_{+} - Y_{*}\|_{F} & = \|{\cal C}_{\lambda}\left(Y + \tau Y^{-1} - S\right) - {\cal C}_{\lambda}\left(Y_{*} + \tau Y_{*}^{-1} - S\right)\|_{F}, \nonumber \\
& \leq \|\left(Y + \tau Y^{-1}\right) - \left(Y_{*} + \tau Y_{*}^{-1} \right)\|_{F}, 
\end{alignat}
To bound the later expression we follow the arguments from \cite[Lemma 3]{brajarat2012nips} and define a function
\begin{alignat}{1}
h_{\tau}(Y) &= \vect(Y) + \vect(\tau Y^{-1}). \nonumber
\end{alignat}
Using $h_{\tau}$, (\ref{above2}) can be written as 
\begin{alignat}{1}
\|\left(Y + \tau Y^{-1}\right) - \left(Y_{*} + \tau Y_{*}^{-1} \right)\|_{F} &= \|h_{\tau}(Y) - h_{\tau}(Y_{*})\|_{2}. \nonumber
\end{alignat}
The Jacobian $J_{h_{\tau}}$ of $h_{\tau}$ is given by 
\begin{alignat}{1}
J_{h_{\tau}}(Y) &= I_{p^{2}} - \tau Y^{-1} \otimes Y^{-1}. \nonumber
\end{alignat}
The function $h_{\tau}(Y)$ is a differentiable on ${\cal D}$. Therefore applying the mean value theorem gives 
\begin{alignat}{1}
\|h_{\tau}(Y) - h_{\tau}(Y_{*})\|_{2} & \leq \sup_{\delta \in [0,1]} \left\{\|I_{p^{2}} - \tau Y_{\delta}^{-1} \otimes Y_{\delta}^{-1}\|_{2}\right\}\|\vect(Y) - \vect(Y_{*})\|_{2}, \nonumber \\ 
& \leq \sup_{\delta \in [0,1]} \left\{\|I_{p^{2}} - \tau Y_{\delta}^{-1} \otimes Y_{\delta}^{-1}\|_{2}\right\}\| Y - Y_{*}\|_{F}, \nonumber
\end{alignat}
where $Y_{\delta} = \delta Y + (1-\delta)Y_{*}$ is some convex combination of $Y$ and $Y_{*}$. The eigenvalues of $Y_{\delta}$ are bounded using Weyl's inequality which gives
\begin{alignat}{1}
\alpha I \preceq Y_{\delta} \preceq \beta I. \nonumber
\end{alignat}
The eigenvalues of $I_{p^{2}} - \tau Y_{\delta}^{-1} \otimes Y_{\delta}^{-1}$ are $1 - \tau \left[\text{eig}(Y_{\delta})^{-2}\right]$ and therefore the $\sup$ term can be bounded as 
\begin{alignat}{1}
\sup_{\delta \in [0,1]} \left\{\|I_{p^{2}} - \tau Y_{\delta}^{-1} \otimes Y_{\delta}^{-1}\|_{2}\right\} & = \displaystyle\max\left\{\bigg|1-\frac{\tau}{\alpha^{2}}\bigg|, \bigg|1-\frac{\tau}{\beta^{2}}\bigg|\right\}, \nonumber
\end{alignat}
thereby completing the proof.
\end{proof}

\subsection{Proof of Lemma \ref{upperbound}}
\begin{proof}
We begin with the maximum eigenvalue of $Y_{k+1}$. 
\begin{alignat}{1}
\lambda_{\max}(Y_{k+1}) &= \|Y_{k+1}\|_{2} = \left(\|Y_{k+1}\|_{2} - \|Y_{*}\|_{2} \right) + \|Y_{*}\|_{2} \nonumber \\
& \leq \|Y_{k+1} - Y_{*}\|_{2} + \|Y_{*}\|_{2} \nonumber \\
& \leq \|Y_{k+1} - Y_{*}\|_{F} + \|Y_{*}\|_{2} \nonumber
\end{alignat}
Since $\tau_{k} < \alpha_{k}^{2} \leq \beta^{2}$, 
\begin{alignat}{1}
\displaystyle\max\left\{\bigg|1-\frac{\tau_{k}}{\alpha_{k}^{2}}\bigg|, \bigg|1-\frac{\tau_{k}}{\beta^{2}}\bigg|\right\} \leq 1, \nonumber
\end{alignat} 
and therefore,
\begin{alignat}{1}
\lambda_{\max}(Y_{k+1}) & \leq \|Y_{k} - Y_{*}\|_{F} + \|Y_{*}\|_{2}, \nonumber
\end{alignat}
and applying Lemma \ref{strongconvexlemma} repeatedly for $\|Y_{k} - Y_{*}\|_{F} \leq \|Y_{k-1} - Y_{*}\|_{F}$, we have
\begin{alignat}{1}
\lambda_{\max}(Y_{k+1}) \leq \|Y_{0} - Y_{*}\|_{F} + \|Y_{*}\|_{2}, \nonumber
\end{alignat}
which proves the required result with inductive argument on $k$.
\end{proof}

\subsection{Proof of Lemma \ref{lowerbound}}
\begin{proof}
Denote by $a = \lambda_{\min}(Y)$ the smallest eigenvalue of $Y$. By Lemma \ref{upperbound} we have
\begin{alignat}{1}
\log\det Y_{0} < \log\det Y \leq \log(a) + (p-1)\log(\beta) \nonumber
\end{alignat}
For the initial point $Y_{0} = S+\lambda I$, the minimum eigenvalue of $Y_{0}$ satisfies $\lambda \leq \lambda_{\min}(Y_{0})$ and $p\lambda \leq \log\det Y_{0}$. Therefore,
\begin{alignat}{1}
\lambda_{\min}(Y) = a > \left(\det Y_{0}\right)^{p}\beta^{-(p-1)} \geq \lambda^{p}\beta^{-(p-1)}. \nonumber
\end{alignat}
Hence proved.
\end{proof}

\subsection{Proof of Theorem \ref{t:gama}}
\begin{proof}
The iterations in Algorithm \ref{a:gama} can be rewritten by substituting $X_{k+1}$ and $Z_{k+1}$ in terms of $Y_{k}$ as
\begin{alignat}{1}
Y_{k+1} &= Y_{k} + \tau_{k}\left( Y_{k}^{-1} - {\cal S}_{\lambda/\tau}( Y_{k}^{-1} + (Y_{k} -S)/\tau_{k})\right), \nonumber \\
 &= {\cal C}_{\lambda}\left(Y_{k} + \tau Y_{k}^{-1} - S \right) + S. \nonumber
\end{alignat}
The initial point
\begin{alignat}{1}
Y_{0} = S + \lambda I, \nonumber
\end{alignat}
and the subsequent iterates satisfy $\alpha I \preceq Y_{k} \preceq \beta I$ by Lemma \ref{upperbound} and Lemma \ref{lowerbound} for $\alpha, \beta$ defined in the corresponding lemmas. Therefore, using Lemma \ref{strongconvexlemma}, we get equation (\ref{mt1})
\begin{alignat}{1}
\|Y_{k+1} - Y_{*}\|_{F} \leq \displaystyle\max\left\{\bigg|1-\frac{\tau_k}{\alpha^{2}}\bigg|, \bigg|1-\frac{\tau_k}{\beta^{2}}\bigg|\right\} \|Y_{k} - Y_{*}\|_{F}, \nonumber
\end{alignat}
for $k = 0, 1 \ldots$. And the condition on the constant step size $\tau_k < \alpha^2, \forall k$ gives (\ref{mt2}),
\begin{alignat}{1}
\|Y_{k+1} - Y_{*}\|_{F} \leq \gamma \|Y_{k} - Y_{*}\|_{F}, \nonumber
\end{alignat}
thereby proving the linear convergence of Algorithm \ref{a:gama}.
\end{proof}

\subsection{Proof of Theorem \ref{t:gengama}}
\begin{proof}
(a) Using Lemma \ref{proxop}, the modified $Y_{+}$ update (\ref{boundconupdate}) satisfies
\begin{alignat}{1}
\|Y_{+} - Y_{*}\|_{F} & = \|\Pi_{{\cal D}}(Y + \tau Y^{-1}) - \Pi_{{\cal D}}(Y_{*} + \tau Y_{*}^{-1})\|_{F} \nonumber \\
& \leq \|\left(Y + \tau Y^{-1}\right) - \left(Y_{*} + \tau Y_{*}^{-1} \right)\|_{F}, 
\end{alignat}
and the rest of the proof of Lemma \ref{strongconvexlemma} follows accordingly. Similarly the proof of Theorem \ref{t:gama} can easily be adapted to prove the linear convergence of the modified algorithm.\\
(b) To prove the convergence of (\ref{mod2iter}), we use a similar mean-value theorem argument on the function $Y + \tau A^{*}(S + BYA^{*})^{-1}B$, namely,
\begin{alignat}{1}
\|Y_{+} - Y_{*}\|_{F} & \leq \|Y - Y_{*} + \tau A^{*}\left((S + BYA^{*})^{-1} - (S + BY_{*}A^{*})^{-1}\right)B\|_{F} \nonumber \\
& \leq \displaystyle\max\left\{|1-{\tau}{v}|, |1- {\tau}{w}|\right\} \|Y - Y_{*}\|_{F}, 
\end{alignat}
where $v$ and $w$ are constants which depend on $S$, $\rho(A^{*}A)$, $\rho(B^{*}B)$ and $\lambda$.
\end{proof}

\section*{Appendix C: Portfolio Optimization Background}

The formal definitions of the performance metrics used in Section 5.3 as taken from \cite{won2012condition} are as follows
\begin{itemize}\itemsep-2pt
\item {\it Realized return.}
The realized return of a portfolio rebalancing strategy over the entire trading period is
\begin{linenomath*}
\[
r_p = \frac{1}{K}\sum_{j=1}^K \frac{1}{L}\sum^{N_{\text{estim}}+jL}_{t=N_{\text{estim}}+1+(j-1)L}r^{(t) T}w^{(j)}.
\]
\end{linenomath*}
\item {\it Realized risk.}
The realized risk (return standard deviation) of a portfolio rebalancing strategy over the entire trading period is
\begin{linenomath*}
\[
\sigma_p = \sqrt{\frac{1}{K}\sum_{j=1}^K \frac{1}{L}\sum^{N_{\text{estim}}+jL}_{t=N_{\text{estim}}+1+(j-1)L}(r^{(t) T}w^{(j)})^2 - r_p^2}.
\]
\end{linenomath*}
\item {\it Realized Sharpe ratio (SR).}
The realized Sharpe ratio, \ie, the ratio of the excess expected return of a portfolio rebalancing strategy relative to the risk-free return $r_f$ is given by
\begin{linenomath*}
\[
SR = \frac{r_p-r_f}{\sigma_p}.
\]
\end{linenomath*}
\item {\it Turnover.}
The turnover from the portfolio $w^{(j)}$ held at the start date of the $j$th holding period $[N_{\text{estim}}+1+(j-1)L,N_{\text{estim}}+jL]$
 to the portfolio $w^{(j-1)}$ held at the previous period is computed as
\begin{linenomath*}
\[
\text{TO}(j) = \sum^p_{i=1}\left|w^{(j)}_i-\left(\prod^{N_{\text{estim}}+jL}_{t=N_{\text{estim}}+1+(j-1)L}(1+r^{(t)}_{i})\right)w^{(j-1)}_i\right|.
\]
\end{linenomath*}
For the first period, we take $w^{(0)}=0$,
\ie, the initial holdings of the assets are zero.
\item {\it Normalized wealth growth.}
Let $w^{(j)}=(w^{(j)}_1,\ldots,w^{(j)}_n)$ be the portfolio constructed by a rebalancing strategy held over the $j$th holding
period $[N_{\text{estim}}+1+(j-1)L,N_{\text{estim}}+jL]$.
When the initial budget is normalized to one,
the normalized wealth grows according to the recursion
\begin{linenomath*}
\[
W(t) = \left\{\begin{array}{ll}
W(t-1)(1+\sum^p_{i=1}w_{it}r^{(t)}_{i}), &
t \not\in\{N_{\text{estim}}+jL \;|\; j=1,\ldots,K\},
\\
W(t-1)\left((1+\sum^p_{i=1}w_{it}r^{(t)}_{i})-\text{TC}(j)\right), &
t = N_{\text{estim}}+jL,
\end{array}
\right.
\]
\end{linenomath*}
for $t = N_{\text{estim}}, \ldots, N_{\text{estim}}+KL$,
with the initial wealth
$W(N_{\text{estim}}) = 1$.
Here
\begin{linenomath*}
\[
w_{it} = \left\{\begin{array}{cl}
w^{(1)}_i, &
t = N_{\text{estim}}+1, \ldots, N_{\text{estim}}+L,
\\
\vdots
\\
w^{(K)}_i, &
t = N_{\text{estim}}+1+(K-1)L, \ldots, N_{\text{estim}}+KL.
\end{array}
\right.
\]
\end{linenomath*}
and
\begin{linenomath*}
\[
\text{TC}(j) =
\sum^p_{i=1}\eta_i\left|w^{(j)}_i-\left(\prod^{N_{\text{estim}}+jL}_{t=N_{\text{estim}}+1+(j-1)L}(1+r^{(t)}_{i})\right)w^{(j-1)}_i\right|
\]
\end{linenomath*}
is the transaction cost due to the rebalancing if the cost to buy or sell one share of stock $i$ is $\eta_i$.

\item {\it Size of the short side.}
The size of the short side of a portfolio rebalancing strategy 
over the entire trading period is computed as
\begin{linenomath*}
\[
r_p = \frac{1}{K}\sum_{j=1}^K \left(
\sum_{i=1}^p |\min(w_{i}^{(j)},0)| / \sum_{i=1}^p |w_{i}^{(j)}| \right).
\]
\end{linenomath*}
\end{itemize}



\subsection*{Acknowledgment}
We would like to thank Dominique Guillot for useful discussions and suggestions for improving the paper. We would also like to thank the authors of \cite{brajarat2012nips} and \cite{hsieh2011sparse} for providing the code for G-ISTA and QUIC which served as a base for the code development for G-AMA.
\bibliographystyle{plain}
\bibliography{jmlr_gama}

\end{document}